\documentclass[final,12pt]{colt2026} 

\title{Wasserstein Policy Learning for Distributional Outcomes}
\usepackage{times}
\usepackage{mathrsfs}
\usepackage{mathtools}
\usepackage{xcolor}
\newtheorem{assumption}[theorem]{Assumption}

\renewcommand{\thefootnote}{\fnsymbol{footnote}}
\coltauthor{ \Name{Yiyan Huang\protect\footnotemark[1]\protect\footnotemark[2]} \Email{huangyiyan@gbu.edu.cn}\\ \addr School of Computing and Information Technology, Great Bay University, Guangdong, China \AND \Name{Cheuk Hang Leung\protect\footnotemark[1]} \Email{chleung87@cityu.edu.hk}\\ \addr Department of Data Science, City University of Hong Kong, Hong Kong, China \AND \Name{Qi Wu\protect\footnotemark[1]} \Email{ qiwu55@cityu.edu.hk}\\ \addr Department of Data Science, City University of Hong Kong, Hong Kong, China \AND \Name{Zhiheng Zhang\protect\footnotemark[1]} \Email{zhangzhiheng@mail.shufe.edu.cn}\\ \addr School of Statistics and Data Science \& Institute of Big Data Research, Shanghai University of Finance and Economics, Shanghai, China }

\begin{document}
	
	\maketitle
	\footnotetext[1]{Authors are in alphabetical order.} \footnotetext[2]{Corresponding author.} 
	\begin{abstract}%
		Offline policy learning has received growing attention in causal inference. The primary objective is to learn a policy (individualized treatment rule) as a mapping from covariates to treatment that maximizes the empirical welfare defined as the mean of scalar-valued potential outcomes. In this paper, we study offline policy learning with distribution-valued outcomes, where each potential outcome is a probability measure on $\mathbb{R}$ and the reward is defined through a utility functional applied to the Wasserstein barycenter of induced outcome distributions.
		We establish statistical guarantees for the policy learning framework based on both Inverse Probability Weighting (IPW) and Doubly Robust (DR) estimators. By handling the challenging uniform deviation over the product of the combinatorial policy class and the infinite-dimensional quantile domain, we prove that the finite-sample regret has leading dependence $\widetilde{\mathcal{O}}(\sqrt{\mathrm{N\text{-}dim}(\Pi)/N})$. In the one-dimensional Wasserstein setting and under the stated regularity conditions, the leading regret rate is still governed by the policy-class complexity. Moreover, we provide a minimax lower bound establishing the sharpness of the leading dependence on $N$ and $\mathrm{N\text{-}dim}(\Pi)$.
	\end{abstract}

	\begin{keywords}%
		Causal inference, policy learning, distributional outcome
	\end{keywords}

	\section{Introduction}
	
	Offline policy learning, aiming to derive individualized treatment rules from observational data to maximize population-level welfare, is an important approach for personalized decision-making in causal inference \citep{zhao2012estimating, swaminathan2015batch, zhou2017residual, kitagawa2018should, kallus2018balanced, kallus2021minimax, athey2021policy}. In the classical regime, a policy $\pi: \mathcal{X} \to \mathcal{A}$ is evaluated based on the expectation of a scalar potential outcome $Y \in \mathbb{R}$ (e.g., the Average Treatment Effect). This paradigm has been extensively studied, with empirical welfare maximization (EWM) approaches achieving minimax optimal rates by leveraging plug-in estimators such as Inverse Probability Weighting (IPW) \citep{kitagawa2018should} or Doubly Robust (DR) \citep{athey2021policy} methods to construct unbiased surrogates of the policy value.
	
	However, reducing welfare to a scalar expectation fails to capture distributional nuances such as risk, inequality, or tail behavior. This deficiency has motivated the study of \textit{distributional policy learning}, where the objective targets functionals of the outcome distribution (e.g., quantiles or CVaR) \citep{wang2018quantile, lin2023causal, cui2025policy}. Crucially, while these methods optimize distributional criteria, the underlying potential outcome $Y$ remains a scalar random variable. The complexity arises solely from the non-linearity of the objective function, not from the structure of the outcome space itself.
	
	In contrast, many modern applications feature outcomes that are inherently stochastic processes or probability measures. For instance, policymakers may aim to optimize the entire wealth distribution shape to mitigate inequality \citep{RePEc:fip:fedcwq:191800}, or healthcare systems may target physiological dynamics modeled as measures on path space \citep{zhou2025dynamic}. In these settings, the outcome $\mathcal{Y}$ is naturally an element of the Wasserstein space of probability measures, $\mathcal P_2(\mathbb{R})$ equipped with $\mathcal W_2$. Treating such measure-valued outcomes as densities in a linear space (e.g., $L_2$) and applying standard functional averages can destroy the intrinsic geometry of the data. As noted in \cite{panaretos2019statistical}, linear averaging fails to account for mass displacement, often yielding ``barycenters" that do not represent any individual realization in the population (e.g., the linear average of two unimodal distributions may be bimodal).
	
	This necessitates a shift to \textit{policy learning with distributional outcomes}. Our goal is to learn a policy that maximizes a utility defined on the \textit{Wasserstein barycenter}, which is the Fréchet mean of the induced outcome distributions. This introduces a fundamental theoretical hurdle: the Wasserstein barycenter is defined as the solution to an optimization problem over a metric space, lacking the closed-form linearity of expectations \citep{kurisu2024geodesic}. Consequently, the learning problem for distributional outcome cannot be trivially reduced to standard scalar-outcome empirical welfare maximization like \cite{kitagawa2018should, athey2021policy}. The policy in this case induces an entire quantile curve in $L_2([0,1])$, and our objective relies on the Wasserstein barycenter, which preserves the geometry of optimal transport and differs fundamentally from the linear averaging of densities or cumulative distribution functions. Statistically, addressing these issues requires establishing uniform convergence guarantees for policy-dependent objects in an infinite-dimensional space, which necessitates controlling the complex interplay between the combinatorial complexity of the policy class and the metric entropy of the outcome space.
	
	To see why this poses a technical barrier, consider a concrete scenario in precision medicine: using continuous glucose monitoring data to learn an insulin dosing policy. The outcome is not a single value, but a distributional profile of glucose levels over time (a probability measure). This setting involves the following challenges:
	\begin{itemize}
		\item The functional complexity: The glucose profile is an infinite-dimensional object with complex shape features---it may have multiple peaks (post-prandial spikes) or heavy tails (hypoglycemia risks). Capturing this full distributional shape requires covering the massive functional space of potential quantile curves, indexed by $t \in [0,1]$.
		\item The policy complexity: Simultaneously, the policy searches through a vast combinatorial space of decision rules (e.g., deep decision trees based on genomic data) to find the optimal subgroup assignment.
		\item The complex interplay: The true technical barrier arises from the exploitation of functional flexibility by the combinatorial search. A policy learner, in its pursuit of empirical welfare, may ``cheat'' by overfitting to statistical noise at specific quantile levels—for instance, artificially overfitting the median quantile level ($t=0.5$) to maximize utility while unknowingly destabilizing the tails ($t \to 0$ or $1$), potentially leading to clinically undesirable tail behavior.
	\end{itemize}
	This example demonstrates that while policy learning with distributional outcomes is essential, it presents unique theoretical hurdles.
Unlike scalar policy learning, where the risk of misestimation is limited to a single expected value, distribution-valued outcomes introduce policy-indexed objects over an infinite-dimensional quantile domain. The main statistical question is whether replacing scalar outcomes by distribution-valued outcomes introduces an additional nonparametric price in policy learning. Our analysis shows that, in the one-dimensional Wasserstein setting, after the quantile-isometry reduction, the leading regret remains governed by the policy-class complexity. Technically, this requires a product-index uniform deviation bound over the policy class $\Pi$ and the quantile index $t\in[0,1]$.
	
In this paper, we address these challenges by developing a rigorous framework for offline policy learning with distribution-valued outcomes. We formulate the problem by combining a policy-induced $W_2$-barycenter target with a Wasserstein-Lipschitz utility functional. Our key methodological insight leverages the quantile isometry between $(\mathcal P_2(\mathbb{R}),\mathcal W_2)$ and the $L_2$ space of quantile functions. This transformation maps the non-linear barycenter problem into a tractable estimation of policy-indexed quantile curves without sacrificing geometric fidelity. To solve this, we construct Inverse Propensity Weighting (IPW) and cross-fitted Doubly Robust (DR) estimators. Theoretically, the central question is to identify the statistical price of moving from scalar outcomes to distribution-valued outcomes. 
	\paragraph{Organization.} 	Our analysis proceeds in four steps. 
	First, we define a population-level policy value through the Wasserstein barycenter of policy-induced distribution-valued outcomes. 
	Second, in one dimension, we use the quantile isometry to convert barycenter learning into estimation of a policy-indexed mean quantile curve. 
	Third, we construct IPW and cross-fitted DR estimators and enforce validity by monotone rearrangement. 
	Fourth, we prove finite-sample regret upper bounds and a minimax lower bound matching the leading dependence on the sample size and policy-class complexity.

	\paragraph{Contributions.} Our contributions are summarized as follows:
	\begin{itemize}
		\item We formalize the offline policy learning problem where outcomes are probability measures in $\mathcal P_2(\mathbb R)$. By defining the policy value via the Wasserstein barycenter, we generalize standard welfare maximization to respect the intrinsic geometry of distributional data, distinct from functional definitions in Hilbert spaces.
		\item Leveraging the 2-Wasserstein quantile isometry, we reduce barycenter estimation to learning policy-indexed quantile curves. We propose IPW and cross-fitted DR estimators that operate directly on the space of quantile functions. Crucially, we incorporate a monotone rearrangement step that enforces the validity of the estimated quantiles without worsening the $L_2$ estimation error.
		\item We establish finite-sample regret bounds that identify whether distribution-valued outcomes create an additional leading-order price for policy learning. By controlling the product-index uniform deviation over $\Pi\times[0,1]$, we prove a $\widetilde{\mathcal{O}}(\sqrt{\mathrm{N\text{-}dim}(\Pi)/N})$ regret rate. Thus, in the one-dimensional Wasserstein setting, the leading statistical complexity is still governed by the policy class.
		\item We derive minimax lower bounds showing that the leading dependence on sample size $N$ and policy complexity $\mathrm{N\text{-}dim}(\Pi)$ is rate-sharp. Hence the leading rate is sharp, and the quantile-isometry reduction does not hide an additional leading-order nonparametric penalty.
	\end{itemize}
	
	\section{Related Works}
	
	\paragraph{Policy learning in causal inference.}
	A central theoretical goal for offline policy learning is the establishment of minimax regret bounds that scale with the complexity of the policy class. Empirical Welfare Maximization (EWM), relying on IPW or DR estimators, has been the dominant paradigm for scalar outcomes \citep{manski2004statistical, kitagawa2018should, athey2021policy}. For instance, \cite{kitagawa2018should} derived $\mathcal{O}(N^{-1/2})$ regret bounds dependent on the Vapnik-Chervonenkis (VC) dimension of the policy class, while \cite{athey2021policy} extended these guarantees to observational settings using cross-fitting and orthogonal scores. This literature has expanded to cover continuous treatments \citep{kallus2018policy, chernozhukov2019semi, ai2026data}, policy learning without overlap \citep{khan2023off, zhao2024positivity, jin2025policy}, unmeasured confounding \citep{kallus2018confounding, kallus2021minimax}, distribution shift \citep{mo2021learning, kallus2022doubly, mu2022factored, adjaho2022externally, kido2022distributionally, si2023distributionally, shen2024wasserstein}, and adaptive settings \citep{bibaut2021risk, zhan2024policy}. However, the theoretical machinery in these works critically relies on the outcome residing in a low-dimensional Euclidean space, leaving the regime of distribution-valued outcomes unexplored.
	
	\paragraph{Policy learning with distributional objectives.}
	A growing body of work has moved beyond average welfare to optimize distributional functionals, such as quantiles, CVaR, or Gini coefficients \citep{wang2018quantile, kock2024regularizing, cui2025policy, manski2023statistical}. While these methods capture risk and inequality, they represent distributional objectives on scalar outcomes, not distributional outcomes. In these settings, the potential outcome $Y$ remains a scalar random variable, and the challenge lies in the non-linearity of the utility function. Consequently, these approaches typically rely on sorting or rank-based statistics of scalar variables. In contrast, our work addresses outcomes that are themselves probability measures in the 2-Wasserstein space $(\mathcal P_2(\mathbb R),\mathcal W_2)$. This shifts the problem from scalar ranking to metric space optimization. Unlike \cite{cui2025policy} which targets functionals of a scalar $Y$, we integrate Wasserstein barycenters into offline policy learning, necessitating tools that respect the mass-displacement geometry of the outcome itself, and more importantly, derive finite-sample regret bounds for empirical utility maximization, together with minimax lower bounds matching the leading dependence on the sample size and policy-class complexity.
	
	\paragraph{Causal inference with non-scalar outcomes.} Analysis of complex outcomes broadly falls into two paradigms: Hilbert-space approaches and metric-space approaches. The former, including Functional ATE (FATE) methods \citep{ecker2024causal, testa2025doubly}, treat outcomes as elements of $L_2$ and apply linear averaging. While computationally tractable, linear averaging is inadequate for distributions with phase variation, often yielding unrepresentative barycenters \citep{kurisu2024geodesic}. The latter paradigm, Geodesic ATE (GATE), addresses this by estimating treatment effects via Fréchet means in metric spaces like Wasserstein space \citep{lin2023causal, kurisu2024geodesic, bhattacharjee2025doubly, raykov2025kernel}. Crucially, our work differs from the GATE literature in its goal and theoretical scope. Existing GATE works focus on estimation and inference to establish the consistency or asymptotic normality of a treatment effect estimator. In contrast, we address the policy learning problem, which requires selecting the optimal policy from a policy class $\Pi$. This transition shifts the theoretical challenge to deriving a regret bound with uniform convergence rates over the combinatorial complexity of $\Pi$ (e.g., Natarajan dimension). To our knowledge, we provide the first finite-sample regret bounds for decision-making with distributional outcomes in Wasserstein space.
	
	\section{Policy Learning with Distributional Outcome}\label{sec:setup}
	
	\subsection{Notations and Assumptions}
	Let $\mathcal A=\{a_1,\dots,a_d\}$ be a finite action set and $\mathcal X\subseteq\mathbb R^{K}$ be a compact context space. The outcome is \emph{distribution-valued}: $\mathcal Y$ is a random probability measure supported on a compact interval $\mathcal I=[\mathcal I_{\mathrm{low}},\mathcal I_{\mathrm{up}}]\subset\mathbb R$, with finite second moment. We write $\mathcal Y\in\mathcal P_2(\mathbb R)$ and assume $\mathrm{supp}(\mathcal Y)\subseteq \mathcal I$ almost surely. Equivalently, $\mathcal Y$ is represented by its (left-continuous) quantile function $\mathcal Y^{-1}:[0,1]\to\mathcal I$. We observe i.i.d.\ logged data $\{(X_i,A_i,\mathcal Y_i)\}_{i=1}^N$, where $A_i$ is drawn from a behavior policy with propensity $f_0(a|x):=\mathbb P(A=a\mid X=x)$. We adopt the standard potential outcome notation: for each $a\in\mathcal A$ there exists a potential distribution $\mathcal Y[a]\in\mathcal P_2(\mathbb R)$, with quantile function $\mathcal Y[a]^{-1}$, and the observed outcome satisfies $\mathcal Y=\mathcal Y[A]$ almost surely. We state necessary assumptions as follows.
	
	\paragraph{Assumptions.}
	We impose standard causal identification assumptions, adapted to the policy learning framework with distribution-valued outcomes.
	
	\begin{assumption}[Consistency]\label{ass:consistency}
		If $A=a$, then $\mathcal Y=\mathcal Y[a]$ almost surely.
	\end{assumption}
	
	\begin{assumption}[Unconfoundedness]\label{ass:unconf}
		For every $a\in\mathcal A$, $\mathcal Y[a]\perp\!\!\!\perp A \mid X$.
	\end{assumption}
	
	\begin{assumption}[Overlap]\label{ass:overlap}
		There exists $\underline f>0$ such that for all $x\in\mathcal X$ and $a\in\mathcal A$,
		$f_0(a|x)\ge \underline f$.
	\end{assumption}
	
	\begin{assumption}[Boundedness]\label{ass:boundedness}
		There exist continuous functions $\underline q,\overline q:[0,1]\to\mathbb R$ such that, for all $a\in\mathcal A$ and $t\in[0,1]$, almost surely,
		\[
		\underline q(t)\le \mathcal Y[a]^{-1}(t)\le \overline q(t).
		\]
		By consistency, the same bounds hold for the observed quantile $\mathcal Y^{-1}$.
		We denote a uniform bound by $M:=\sup_{t\in[0,1]} \max\{|\underline q(t)|,|\overline q(t)|\}<\infty$.
	\end{assumption}
	
	\begin{assumption}[Quantile regularity in $t$]\label{ass:quantile_lipschitz}
		There exists $\mathcal U<\infty$ such that for all $a\in\mathcal A$ and all $s,t\in[0,1]$,
		\[
		\big|\mathcal Y[a]^{-1}(t)-\mathcal Y[a]^{-1}(s)\big|\le \mathcal U |t-s|
		\qquad \text{almost surely.}
		\]
		By consistency, the same holds for the observed quantile $\mathcal Y^{-1}$.
	\end{assumption}
	
	Assumptions~\ref{ass:consistency}--\ref{ass:overlap} are standard for off-policy evaluation and policy learning.
	Assumptions~\ref{ass:boundedness}--\ref{ass:quantile_lipschitz} control the range and the $t$-regularity of the quantile curves.
	They enable a discretization argument that yields uniform control over the continuum index $t\in[0,1]$ in our theoretical analysis.

	\subsection{The Wasserstein Barycenter Objective}
	Let $\mathcal P_2(\mathbb R)$ denote the space of probability measures on $\mathbb R$ with finite second moments.
	Let $\Pi$ be a class of deterministic policies $\pi:\mathcal X\to\mathcal A$.
	For a given policy $\pi$, the counterfactual distributional outcome is a random measure $\mathcal Y[\pi(X)]$ taking values in $\mathcal P_2(\mathbb R)$.
	We define the policy-induced target as the $W_2$-barycenter of these counterfactual measures:
	\begin{equation}\label{eq:barycenter_def_main_text}
		\mu(\pi)
		\in
		\underset{\mu\in\mathcal P_2(\mathbb R)}{\arg\min} \;
		\mathbb E\Big[\mathcal{W}_2^2\big(\mu,\mathcal Y[\pi(X)]\big)\Big].
	\end{equation}
	This target is the Fr\'echet mean in $(\mathcal P_2(\mathbb R),\mathcal{W}_2)$: it summarizes the policy’s effect
	in a way that respects the geometry of distributions under optimal transport.
	
	The barycenter $\mu(\pi)$ provides a geometry-aware notion of the ``average'' distribution induced by $\pi$.
	To encode preferences over distributions, we optimize a distributional utility $U:\mathcal P_2(\mathbb R)\to\mathbb R$ as welfare for distribution-valued outcome and define
	\[
	\pi^\star\in\underset{\pi\in\Pi}{\arg\max}\;U\big(\mu(\pi)\big).
	\]
	This aggregate criterion should be distinguished from the average individually scalarized objective $\mathbb E[U(\mathcal Y[\pi(X)])]$. The latter first applies $U$ to each unit-level counterfactual distribution and then averages the resulting scalars, whereas $U(\mu(\pi))$ first aggregates the policy-induced distribution-valued outcomes through their Wasserstein barycenter and then evaluates the resulting population-level distributional profile. These two criteria generally differ when $U$ is nonlinear. For example, if half of the population has the degenerate distribution $\delta_0$ and the other half has $\delta_2$, their one-dimensional $W_2$-barycenter is $\delta_1$. For the bounded-support Wasserstein-Lipschitz utility $U(\nu)=-(\int z\,d\nu(z)-1)^2$, one has $U(\delta_1)=0$ but $\frac12U(\delta_0)+\frac12U(\delta_2)=-1$. Thus our objective evaluates the aggregate distributional profile induced by a policy.
	Therefore, we can measure the policy performance by the regret
	\[
	\mathcal R(\pi):=U(\mu(\pi^\star))-U(\mu(\pi)).
	\]
	
	In one dimension, Wasserstein geometry admits a convenient representation through quantile functions. For a measure $\nu\in\mathcal P_2(\mathbb R)$, let $\nu^{-1}:[0,1]\to\mathbb R$ denote its (left-continuous) quantile function. The following proposition highlights that the $\mathcal{W}_2$ metric is isometric to the $L_2([0,1])$ distance between quantile functions.
	\begin{proposition}[Quantile representation of $\mathcal{W}_{p}$ \citep{santambrogio2015optimal}]\label{prop:Quantile representation of Wasserstein} We define $\mu_{1},\mu_{2}$ as two probability measures on $\mathbb R$ with quantile functions
		$\mu_1^{-1},\mu_2^{-1}$. Then for $p\ge 1$,
		\[
		\mathcal{W}_{p}(\mu_{1},\mu_{2})
		=
		\left(\int_{0}^{1}\left|\mu_{1}^{-1}(t)-\mu_{2}^{-1}(t)\right|^{p}\,dt\right)^{\frac{1}{p}}.
		\]
	\end{proposition}
	A direct consequence of the isometry in Proposition~\ref{prop:Quantile representation of Wasserstein} is that the Fr\'echet mean in the Wasserstein space corresponds to the standard Euclidean mean in the quantile space. This yields an explicit characterization of the objective in \eqref{eq:barycenter_def_main_text}, as demonstrated in the following Proposition.
	
	\begin{proposition}[Characterization of the Barycenter]\label{prop:barycenter_quantile_main}
		Under Assumption~\ref{ass:boundedness}, the barycenter $\mu(\pi)$ is unique and satisfies
		\[
		\mu(\pi)^{-1}(t)=\mathbb E\big[\mathcal Y[\pi(X)]^{-1}(t)\big],\qquad \forall t\in[0,1].
		\]
		The proof is given in Appendix~\ref{app:barycenter}.
	\end{proposition}
	
	Proposition~\ref{prop:barycenter_quantile_main} shows that in the one-dimensional case, the barycenter geometry reduces to a mean quantile curve.
	Importantly, this does \emph{not} collapse the learning problem to scalar policy learning:
	the policy-induced object is the entire function
	\[
	q_\pi(t):=\mu(\pi)^{-1}(t)=\mathbb E[\mathcal Y[\pi(X)]^{-1}(t)] \in L_2([0,1]),
	\]
	and our regret analysis requires \emph{uniform learning over the product class} $\Pi\times[0,1]$. Moreover, we assume the utility function $U$ is $L_U$-Lipschitz with respect to $\mathcal W_2$:
	\begin{assumption}[Lipschitz utility]\label{ass:lipschitz_U_main}
		There exists $0<L_U<\infty$ such that for all $\nu_1,\nu_2\in\mathcal P_2(\mathbb R)$,
		\[
		|U(\nu_1)-U(\nu_2)|\le L_U\,\mathcal{W}_2(\nu_1,\nu_2).
		\]
	\end{assumption}
	This includes several common distributional objectives and ensures stability: small Wasserstein errors in $\mu(\pi)$ translate into small utility errors.
	By Assumption~\ref{ass:lipschitz_U_main} and Proposition~\ref{prop:barycenter_quantile_main}, we have
	\begin{equation}\label{eq:regret_by_quantile_main}
		\mathcal R(\pi)
		\le L_U \,\mathcal{W}_2(\mu(\pi^\star),\mu(\pi))
		= L_U\left(\int_0^1|q_{\pi^\star}(t)-q_\pi(t)|^2\,dt\right)^{1/2}.
	\end{equation}
	
	\subsection{Policy-class Complexity: Natarajan Dimension}
	Since $\Pi$ is multi-class, we quantify its complexity via the Natarajan dimension.
	
	\begin{definition}[Natarajan dimension]\label{def:natarajan_main}
		A set $\{x_1,\dots,x_r\}\subseteq\mathcal X$ is \emph{Natarajan-shattered} by $\Pi$ if there exist
		$f_1,f_2:\{x_1,\dots,x_r\}\to\mathcal A$ such that (i) $f_1(x_j)\neq f_2(x_j)$ for all $j$ and
		(ii) for every $S\subseteq[r]$ there exists $\pi\in\Pi$ with
		$\pi(x_j)=f_1(x_j)$ for $j\in S$ and $\pi(x_j)=f_2(x_j)$ for $j\notin S$.
		The Natarajan dimension $V:=\mathrm{N\text{-}dim}(\Pi)$ is the maximum $r$ such that some set of
		size $r$ is shattered.
	\end{definition}
	
	\section{Main Results}\label{sec:estimators}
	
	\noindent \textbf{Identification.} We study statistical guarantees for learning $\pi$ by maximizing an empirical utility $U(\hat\mu(\pi))$.
	By Proposition~\ref{prop:barycenter_quantile_main}, the population target is the mean quantile curve
	$q_\pi(t)=\mu(\pi)^{-1}(t)=\mathbb E[\mathcal Y[\pi(X)]^{-1}(t)]$.
	Under Assumptions~\ref{ass:consistency}--\ref{ass:overlap}, $q_\pi(t)$ admits standard identification formulas. Define the outcome regression model
	\[
	m_0(a,x)(t):=\mathbb E[\mathcal Y^{-1}(t)\mid A=a,X=x],
	\]
	which equals $\mathbb E[\mathcal Y[a]^{-1}(t)\mid X=x]$ by unconfoundedness.
	Then for each fixed $(\pi,t)$,
	\[
	q_\pi(t)=\mathbb E\!\left[\frac{\mathbf 1_{\{A=\pi(X)\}}\,\mathcal Y^{-1}(t)}{f_0(A|X)}\right]
	=\mathbb E\big[m_0(\pi(X),X)(t)\big].
	\]
	This identification is useful for constructing IPW and DR estimators.
	\paragraph{Statistical barrier: bound over $\Pi \times [0,1]$.} The central difficulty lies in establishing concentration inequalities that hold simultaneously over both the complex policy class $\pi \in \Pi$ and the functional index $t \in [0,1]$. Unlike scalar policy learning where the target is a single value, here the index set is the product of a multi-class hypothesis space and a continuum domain. 
	
	\paragraph{Monotone rearrangement for $\hat\mu(\pi)$.}
	In practice, a raw estimator $\hat{q}_\pi$ (e.g., via IPW or DR) is not guaranteed to be non-decreasing, and thus may fail to be a valid quantile function. To ensure that the induced measure $\hat{\mu}(\pi)$ is well-defined, we employ the standard monotone rearrangement operator. Specifically, let $\mathcal{Q}$ denote the set of left-continuous, non-decreasing functions mapping $[0, 1]$ to $\mathbb{R}$. We define the calibrated quantile function as:
	$\hat{q}_\pi^\uparrow := \Pi_{\mathcal{Q}}(\hat{q}_\pi)$, 
	where $\Pi_{\mathcal{Q}}$ is the projection that finds the nearest element in $\mathcal{Q}$ with respect to a suitable norm. The estimated measure $\hat{\mu}(\pi)$ is then uniquely defined as the distribution whose quantile function is $\hat{q}_\pi^\uparrow$. Crucially, such rearrangement operators are known to be non-expansive contractions and thus do not increase the estimation error and can be treated as a standard technicality in the subsequent regret analysis.
	
	\subsection{Regret Upper Bound for IPW Formulation}\label{sec:Policy Learning IPW}
	
	We now formally define the Inverse Propensity Weighting (IPW) estimator for the distributional setting. For a fixed policy $\pi$ and any quantile level $t\in[0,1]$, Proposition~\ref{prop:barycenter_quantile_main} establishes that the target barycenter satisfies $q_\pi(t)=\mathbb E[\mathcal Y[\pi(X)]^{-1}(t)]$. Under the assumption of unconfoundedness and strict overlap, we can identify this quantity from observational data via the inverse-propensity representation:
	\[
	q_\pi(t)=\mathbb E\Big[\frac{\mathbf 1\{A=\pi(X)\}}{f_0(A|X)}\ \mathcal Y^{-1}(t)\Big].
	\]
	The empirical IPW estimator for the quantile curve is defined as
	\[
	\hat q^{\mathrm{IPW}}_\pi(t)
	:=
	\mathbb P_N\!\left[\frac{\mathbf 1_{\{A=\pi(X)\}}\mathcal Y^{-1}(t)}{f_{0}(A|X)}\right].
	\]
	While $\hat q^{\mathrm{IPW}}_\pi(t)$ is unbiased pointwise, it is not guaranteed to be monotonically increasing with respect to $t$, which violates the definition of a valid quantile function. To enforce validity, we apply monotone rearrangement, denoted by the projection operator $\Pi_{\mathcal Q}$. Let $\hat q^{\mathrm{IPW},\uparrow}_\pi := \Pi_{\mathcal Q}(\hat q^{\mathrm{IPW}}_\pi)$ be the projected estimator. We then define the estimated policy-induced barycenter $\hat\mu^{\mathrm{IPW}}(\pi)$ as the measure whose quantile function corresponds to this valid curve, i.e., $(\hat\mu^{\mathrm{IPW}}(\pi))^{-1} = \hat q^{\mathrm{IPW},\uparrow}_\pi$. Finally, the optimal policy is estimated by maximizing the empirical utility over the policy class $\Pi$:
	\[
	\hat\pi^{\mathrm{IPW}}\in\underset{\pi\in\Pi}{\arg\max}\;U\big(\hat{\mu}^{\mathrm{IPW}}(\pi)\big).
	\]
	
	\begin{theorem}\label{thm:ipw_main}
		Fix a confidence level $\delta\in(0,1)$ and a uniform grid $\mathcal T_{\mathrm{par}}=\big\{0=t_{0}<t_{1}<\cdots<t_{\mathcal{J}}=1:t_{j}=\frac{j}{\mathcal{J}},\;0\leq j\leq\mathcal{J}\big\}$ with mesh size $\eta:=1/\mathcal J$.
		Suppose Assumptions~\ref{ass:consistency}--\ref{ass:quantile_lipschitz} hold, and let the policy class complexity be bounded by $N\ge \mathrm{N\text{-}dim}(\Pi) \ge 1$.
		Furthermore, let $d=|\mathcal A|$ denote the number of actions, $\underline{f}$ be the overlap lower bound, $M$ be the uniform bound on the outcome quantiles, $\mathcal{U}$ be the Lipschitz constant of the quantile curves with respect to $t$, and $L_U$ be the Lipschitz constant of the utility functional $U$.
		Then, with probability at least $1-\delta$, the regret of the IPW learned policy satisfies:
		\[
		\mathcal{R}(\hat{\pi}^{\mathrm{IPW}})
		\le
		\frac{2L_{U}M}{\underline f}\left(
		\sqrt{\frac{2\mathrm{N\text{-}dim}(\Pi)\log(e\cdot N\cdot d)}{N}}
		+
		\sqrt{\frac{2\log\left(\frac{2(\mathcal{J}+1)}{\delta}\right)}{N}}
		\right)
		+
		\left(\frac{4L_U \cdot \eta 
			\cdot \mathcal U}{\underline f}\right).
		\]
	\end{theorem}
	
	\paragraph{Proof sketch.}
	The proof is detailed in Appendix~\ref{app:proof_ipw}. The proof idea proceeds in two main steps. 
	
	(1) Reduction to uniform deviation via geometry.
	We first relate the regret $\mathcal{R}(\hat{\pi}^{\mathrm{IPW}})$ to the estimation error of the barycenter. Using the Lipschitz property of the utility $U$ in Assumption \ref{ass:lipschitz_U_main} and the argmax definition of the estimator, we bound the regret by the worst-case Wasserstein distance  $\sup_{\pi}\mathcal{W}_2(\hat\mu^{\mathrm{IPW}}(\pi),\mu(\pi))$. Crucially, utilizing the Wasserstein-quantile isometry and the non-expansive property of the monotone projection operator, we reduce this geometric error to the uniform deviation of the raw quantile curves in the $L_\infty$ norm: $\sup_{\pi}\sup_t|\hat q_\pi^{\mathrm{IPW}}(t)-q_\pi(t)|$.
	
	(2) Combinatorial complexity.
	Since the index $t$ is continuous, a direct union bound is infeasible. We employ a covering argument by discretizing $t$ onto a grid $\mathcal T_{\mathrm{par}}$. By the Lipschitz continuity of the outcome quantiles (Assumption~\ref{ass:quantile_lipschitz}), we control the approximation error between grid points. On the grid, the complexity of the function class is driven by the policy search. We control this complexity using the Natarajan dimension of $\Pi$, applying a multi-class Sauer's lemma to bound the growth function and Hoeffding’s inequality to bound the point-wise deviations.
	
	Finally, we combine the statistical estimation error (governed by the sample size $N$ and policy complexity $\mathrm{N\text{-}dim}(\Pi)$) with the deterministic discretization error (governed by the grid mesh $\eta$). The resulting bound captures the trade-off between grid resolution and statistical variance.
	
	\paragraph{Interpretation.}
	The upper bound in Theorem~\ref{thm:ipw_main} reveals the statistical nature of learning with distributional outcomes. The first term represents the stochastic estimation error, which scales as $\widetilde{\mathcal O}(N^{-1/2})$ up to logarithmic factors. This term depends on the difficulty of the policy search (measured by $\mathrm{N\text{-}dim}(\Pi)$), the overlap ($\underline{f}^{-1}$), and the number of grid points ($\log \mathcal{J}$). Notably, the dependence on the functional dimensionality (represented by the grid size $\mathcal{J}$) is only logarithmic, indicating that the infinite-dimensional nature of the outcome does not incur a polynomial penalty in sample complexity.
	The second term represents the approximation bias $\mathcal{O}(\eta)$, arising from discretizing the quantile curves. This highlights a bias-variance trade-off: a finer grid (smaller $\eta$) reduces bias but increases the logarithmic stochastic error. However, by choosing $\mathcal{J} \asymp \sqrt{N}$ and given the definition $\eta = 1/\mathcal{J}$, the bias term becomes negligible ($\mathcal{O}(N^{-1/2})$), allowing the estimator to achieve the parametric rate $\widetilde{\mathcal O}(\sqrt{V/N})$ typical of scalar policy learning, without introducing an additional leading-order nonparametric penalty under the one-dimensional quantile-isometry reduction.

	\subsection{Regret Upper Bound for DR Formulation}\label{sec:Policy Learning DR}
	While the IPW estimator is unbiased, it can be inefficient and may suffer from high variance, especially when propensity scores are close to zero. To mitigate this and reduce sensitivity to nuisance estimation errors, we adopt the cross-fitted Doubly Robust (DR) estimator \citep{chernozhukov2018double}. This estimator augments the IPW objective with a regression-based objective, improving stability through Neyman orthogonality.
	
	We employ $L$-fold cross-fitting to decouple the nuisance estimation from the policy evaluation. The sample is randomly partitioned into $L$ disjoint folds $\mathcal I_1,\dots,\mathcal I_L$ of size $n:=N/L$. For each fold $\ell$, we construct nuisance estimators $\hat f_0^\ell$ (propensity score) and $\hat m_0^\ell$ (conditional outcome quantile) using data from the complement folds $\mathcal I_{-\ell}:=\cup_{j\neq \ell}\mathcal I_j$.
The cross-fitted DR estimator for the quantile curve $q_\pi(t)$ is defined as:
	\[
	\hat q^{\mathrm{DR}}_\pi(t)
	:=
	\frac{1}{L}\sum_{\ell=1}^L
	\mathbb P_{n,\ell}\!\left[
	\hat m_0^\ell(\pi(X),X)(t)
	+\frac{\mathbf 1_{\{A=\pi(X)\}}}{\hat f_0^\ell(A|X)}
	\big(\mathcal Y^{-1}(t)-\hat m_0^\ell(\pi(X),X)(t)\big)
	\right],
	\]
	where $\mathbb P_{n,\ell}$ denotes the empirical average over fold $\mathcal I_\ell$.
	Similar to the IPW case, we apply the monotone rearrangement operator $\Pi_{\mathcal Q}$ to obtain a valid quantile curve $\hat q^{\mathrm{DR},\uparrow}_\pi := \Pi_{\mathcal Q}(\hat q^{\mathrm{DR}}_\pi)$, and define the induced barycenter $\hat\mu^{\mathrm{DR}}(\pi)$ via $(\hat\mu^{\mathrm{DR}}(\pi))^{-1} = \hat q^{\mathrm{DR},\uparrow}_\pi$.
	The optimal policy is then learned by maximizing the empirical utility:
	\[
	\hat{\pi}^{\mathrm{DR}}\in\underset{\pi\in\Pi}{\arg\max}\; U(\hat{\mu}^{\mathrm{DR}}(\pi)).
	\]
	Moreover, for each fold $\ell \in \{1,\dots,L\}$, define the uniform estimation errors:
	\begin{equation*}
		\begin{aligned}
			\|\hat f_0^\ell-f_0\|_\infty&:=\sup_{x\in\mathcal X,a\in\mathcal A}\big|\hat f_0^\ell(a|x)-f_0(a|x)\big|,\\
			\|\hat m_0^\ell-m_0\|_{\infty,[0,1]}&:=\underset{\substack{x,a\\t\in[0,1]}}{\sup}\big|\hat m_0^\ell(a,x)(t)-m_0(a,x)(t)\big|.
		\end{aligned}
	\end{equation*}
	Now we provide the main theorem for the regret bound of the cross-fitted DR formulation.

	\begin{theorem}\label{thm:statistical guarantee DR}
		Fix a confidence level $\delta \in (0,1/4)$ and a uniform grid $\mathcal{T}_{\mathrm{par}}$ with mesh size $\eta = 1/\mathcal{J}$.
		Let $d=|\mathcal A|$, $L$ be the number of folds, and $N \ge \mathrm{N\text{-}dim}(\Pi) \ge 1$.
		Assume the outcome is bounded by $M$ and the overlap is bounded by $\underline{f}$.
		Suppose Assumptions~\ref{ass:consistency}--\ref{ass:quantile_lipschitz} and Assumption~\ref{ass:lipschitz_U_main} hold.
		Further, for each fold $\ell$ and every $\gamma\in(0,1)$, assume that with probability at least $1-\gamma$,
		\begin{equation*}
			\begin{aligned}
				\|\hat{m}_{0}^{\ell}-m_{0}\|_{\infty,[0,1]}
				\leq\operatorname{Rate}_{m_{0}}(N,\gamma),
				\qquad
				\|\hat{f}_{0}^{\ell}-f_{0}\|_{\infty}
				\leq \operatorname{Rate}_{f_{0}}(N,\gamma).
			\end{aligned}
		\end{equation*}
		When $\hat f_0^\ell$ appears in an inverse-propensity weight, we maintain the bounded-away-from-zero convention $\hat f_0^\ell(a|x)\ge \underline f$.
		Assume additionally that the fitted outcome-quantile curves are uniformly Lipschitz in $t$, so that, for all folds $\ell$, actions $a$, contexts $x$, and $s,t\in[0,1]$,
		\[
		\left|\left(m_0(a,x)-\hat m_0^\ell(a,x)\right)(t)-\left(m_0(a,x)-\hat m_0^\ell(a,x)\right)(s)\right|
		\le 2\mathcal U |t-s|,
		\]
		with the fitted Lipschitz constant absorbed into $\mathcal U$.
		Define
		\[
		r_f:=\operatorname{Rate}_{f_{0}}\left(N,\frac{\delta}{3L}\right),
		\qquad
		r_m:=\operatorname{Rate}_{m_{0}}\left(N,\frac{\delta}{3L}\right).
		\]
		Without loss of generality, the rate functions are taken to be nonincreasing in the confidence parameter after replacing them by their monotone envelopes. Also define
		\[
		\mathcal{V}_N(\Pi,\delta)
		:=
		\sqrt{\frac{2L\operatorname{N\text{-}dim}(\Pi)\log\left(e\frac{N}{L}d\right)}{N}}
		+
		\sqrt{\frac{2L\log\left(\frac{8\mathcal{J}L}{\delta}\right)}{N}}.
		\]
		Then, with probability at least $1-4\delta$, the regret satisfies
		\begin{equation}
		\label{eq:dr_regret_rearranged}
		\mathcal{R}(\hat{\pi}^{\mathrm{DR}})
		\leq
		C_{\mathrm{or}}\mathcal{V}_N(\Pi,\delta)
		+
		C_{\mathrm{grid}}\eta
		+
		\mathrm{Rem}_{\mathrm{nuis}}(N,\delta,\eta),
		\end{equation}
		where
		\[
		C_{\mathrm{or}}
		:=
		4L_U M\left(1+\frac{2}{\underline f}\right),
		\qquad
C_{\mathrm{grid}}
:=
2L_U\mathcal U\left(4+\frac{6}{\underline f}\right),
		\]
		and the nuisance-induced remainder is
		\begin{align*}
		\mathrm{Rem}_{\mathrm{nuis}}(N,\delta,\eta)
		:=\;&
		2L_U
		\Bigg[
		\left(
		\frac{2r_fr_m}{\underline f^2}
		+
		\frac{4Mr_f}{\underline f^2}
		+
		2\left(1+\frac{1}{\underline f}\right)r_m
		\right)
		\mathcal{V}_N(\Pi,\delta)
		\\
		&\qquad\qquad
		+
		\frac{5r_fr_m}{\underline f^2}
		+
		\frac{4\mathcal U r_f}{\underline f^2}\eta
		\Bigg].
		\end{align*}
		Equivalently, there exists a constant $C>0$, depending only on $L_U,M,\underline f,\mathcal U,L$, and $d$, such that
		\[
		\mathcal R(\hat\pi^{\mathrm{DR}})
		\le
		C\left[
		\mathcal V_N(\Pi,\delta)
		+
		\eta
		+
		r_fr_m
		+
		r_f\mathcal V_N(\Pi,\delta)
		+
		r_m\mathcal V_N(\Pi,\delta)
		+
		r_f\eta
		\right].
		\]
	\end{theorem}
	
	\paragraph{Proof sketch.}
	The proof is detailed in Appendix~\ref{app:proof_DR}.
	The regret is first reduced to the uniform $L_2$ deviation of the estimated policy-indexed quantile curve:
	\[
	\mathcal R(\hat\pi^{\mathrm{DR}})
	\le
	2L_U
	\left(
	\int_0^1
	\left[
	\sup_{\pi\in\Pi}
	\left|
	\hat q_\pi^{\mathrm{DR}}(t)-q_\pi(t)
	\right|
	\right]^2dt
	\right)^{1/2}.
	\]
	The DR score is then decomposed into an oracle empirical-process term, a second-order product bias term, and nuisance-dependent centered empirical-process terms.
	The key point is Neyman orthogonality: the first-order conditional mean terms in the propensity and outcome-regression errors vanish after conditioning on the training folds used to estimate the nuisance functions.
	Consequently, the standalone nuisance bias is of product order $r_fr_m$; the first-order terms $r_f$ and $r_m$ appear only multiplied by the policy-complexity factor $\mathcal V_N(\Pi,\delta)$ or the grid mesh $\eta$.
	
	\paragraph{Rate interpretation and Neyman orthogonality.}
	The bound separates the oracle policy-learning error from the nuisance-induced error. Up to constants depending only on $L_U,M,\underline f,\mathcal U,L$, and $d$, Theorem~\ref{thm:statistical guarantee DR} implies
	\[
	\mathcal R(\hat\pi^{\mathrm{DR}})
	\lesssim
	\mathcal V_N(\Pi,\delta)
	+
	\eta
	+
	r_fr_m
	+
	r_f\mathcal V_N(\Pi,\delta)
	+
	r_m\mathcal V_N(\Pi,\delta)
	+
	r_f\eta.
	\]
	The key point is that Neyman orthogonality removes the standalone first-order nuisance bias. The only non-centered nuisance bias is of product order $r_fr_m$. The terms involving $r_f$ or $r_m$ alone are centered empirical-process or discretization remainders and are multiplied by $\mathcal V_N(\Pi,\delta)$ or $\eta$.
	Consequently, if
	\[
	r_f=\mathcal{O}(N^{-\alpha_f}),
	\qquad
	r_m=\mathcal{O}(N^{-\alpha_m}),
	\qquad
	\mathcal V_N(\Pi,\delta)=\widetilde{\mathcal O}(N^{-1/2}),
	\qquad
	\eta=\mathcal{O}(N^{-1/2}),
	\]
	then
	\[
	\mathcal R(\hat\pi^{\mathrm{DR}})
	=
	\widetilde{\mathcal O}\!\left(
	N^{-1/2}
	+
	N^{-(\alpha_f+\alpha_m)}
	+
	N^{-(\alpha_f+1/2)}
	+
	N^{-(\alpha_m+1/2)}
	\right).
	\]
	In particular, if both nuisance estimators converge at the $N^{-1/4}$ rate, then
	\[
	\mathcal R(\hat\pi^{\mathrm{DR}})
	=
	\widetilde{\mathcal O}(N^{-1/2}).
	\]
	
	\section{Minimax Lower Bound}\label{sec:lower_bound_main}
	To complement the upper bounds derived in the previous sections, we now establish the fundamental information-theoretic limits of policy learning with distributional outcomes. Our goal is to show that the dependence on the sample size $N$ and the policy complexity $V:=\mathrm{N\text{-}dim}(\Pi)$ in our upper bounds is essentially tight.
	
	We focus on a specific class of distributional utilities, the integrated quantile utility, defined as:
	\[
	U_\alpha(\nu):=\int_0^\alpha \nu^{-1}(t)\,dt,\qquad \alpha\in(0,1].
	\]
	This utility fits within our framework as it satisfies the Lipschitz property (Assumption~\ref{ass:lipschitz_U_main}) via the Cauchy-Schwarz inequality: $|U_\alpha(\nu_1)-U_\alpha(\nu_2)|\le \sqrt{\alpha}\,\mathcal{W}_2(\nu_1,\nu_2)$. By constructing a hard instance for this specific utility, we demonstrate the hardness of the general problem.
	
	\begin{theorem}[Minimax Lower Bound]\label{thm:lower_bound_main}
		Assume the policy class satisfies $\mathrm{N\text{-}dim}(\Pi)\ge 1$, let $d=|\mathcal A|\ge2$, fix $\alpha\in(0,1]$, and take $0<\underline f\le 1/d$.
		Let $q_-,q_+:[0,1]\to\mathcal I$ be two left-continuous, non-decreasing, $\mathcal U$-Lipschitz quantile curves satisfying $q_-(t)\le q_+(t)$ and $\sup_{t\in[0,1]}\max\{|q_-(t)|,|q_+(t)|\}\le M$.
		Define $\Delta_Q:=\int_0^\alpha (q_+(t)-q_-(t))\,dt$, and let $\mathcal{P}_{\text{lower}}(q_-,q_+)$ be the subclass of distributions satisfying Assumptions~\ref{ass:consistency}--\ref{ass:quantile_lipschitz} whose potential-outcome quantile curves take values in $\{q_-,q_+\}$ and whose behavior policy has overlap at least $\underline f$.
		There exists a universal constant $c_0>0$ such that for any learning algorithm that maps a dataset $\mathcal{D}_N$ to a policy $\hat\pi$, the worst-case regret is bounded from below by:
		\[
		\inf_{\hat\pi}\ \sup_{\mathbb P\in\mathcal P_{\text{lower}}(q_-,q_+)}\ 
		\mathbb E_{\mathcal{D}_N \sim \mathbb{P}^{\otimes N}}\!\left[\mathcal{R}(\hat{\pi})\right]
		\ \ge\
		c_0\,\Delta_Q
		\min\left\{1,\sqrt{\frac{\mathrm{N\text{-}dim}(\Pi)}{\underline f\,N}}\right\},
		\] 
		where \(\mathcal{R}(\hat{\pi}) = U_\alpha(\mu_{\mathbb{P}}(\pi_{\mathbb{P}}^{\ast}))-U_\alpha(\mu_{\mathbb{P}}(\hat{\pi}))\) is the regret.
	\end{theorem}
	
	\paragraph{Proof sketch.}
	The detailed proof is provided in Appendix~\ref{app:proof_lower_bound}.
	We employ the method of Assouad's Lemma by constructing a hard instance family parameterized by the vertices of a hypercube $\mathcal V=\{\pm1\}^{\mathrm{N\text{-}dim}(\Pi)}$. By designing the potential outcomes as distribution-valued random mixtures of two valid base quantile curves $q_-$ and $q_+$, we reduce the utility maximization problem to a multiple hypothesis testing problem of identifying the optimal action configuration. We then bound the error probability via the Kullback-Leibler (KL) divergence between distributions induced by adjacent vertices. Since the difficulty of distinguishing the optimal actions under noisy outcomes scales inversely with the sample size $N$ and overlap $\underline{f}$, balancing the regret magnitude with this detection probability yields the leading minimax dependence $\sqrt{\mathrm{N\text{-}dim}(\Pi)/(\underline f N)}$ in the large-sample regime.
	
	\paragraph{Interpretation.}
	Theorem~\ref{thm:lower_bound_main} supports the rate sharpness of the proposed framework with respect to $N$ and policy-class complexity. In the large-sample regime, the lower bound scales as $\sqrt{\mathrm{N\text{-}dim}(\Pi)/(\underline f N)}$, which matches the leading $\mathcal{O}(N^{-1/2})$ dependence of our DR upper bound (Theorem~\ref{thm:statistical guarantee DR}) up to logarithmic factors. The bound clarifies the cost of the distributional structure: while it depends on the integrated gap $\Delta_Q$ (analogous to the scalar outcome range), it does not introduce a leading dependence on the non-parametric metric entropy of the outcome space in this one-dimensional construction.

	\section{Conclusion}
	We formulated offline policy learning when each potential outcome is a probability measure and policy performance is evaluated through a utility of the policy-induced Wasserstein barycenter. Focusing on one-dimensional outcomes, we leveraged the quantile representation of \(W_2\) to express the barycenter as a mean quantile curve, enabling empirical welfare maximization while preserving the optimal-transport geometry. We analyzed two estimators of the mean quantile curve, an IPW estimator with known propensities and a cross-fitted doubly robust estimator with estimated nuisance functions.
	Our main theoretical results provide finite-sample regret guarantees with leading dependence \(\widetilde{\mathcal O}(\sqrt{\mathrm{N\text{-}dim}(\Pi)/N})\) on the Natarajan dimension of the policy class and the sample size. We also established a minimax lower bound showing sharpness of the leading dependence on \((\mathrm{N\text{-}dim}(\Pi),N)\) for a representative one-dimensional construction.
	
	Several directions remain open. This work mainly focuses on one-dimensional distributional outcomes, where the quantile isometry provides an exact representation of Wasserstein geometry and makes finite-sample policy regret analysis tractable. Extending the analysis to multivariate Wasserstein spaces is substantially more challenging, since there is no canonical quantile ordering and Wasserstein barycenters generally lack the explicit structure used in our proofs. Another important direction is to relax the quantile regularity conditions imposed here, for example by allowing weaker smoothness or tail behavior. Finally, stochastic policies, continuous treatments, and sequential decision problems involve different policy classes and complexity measures, and therefore require separate regret analyses. We leave these extensions for future work.
	
	\section*{Acknowledgments}
	Yiyan Huang was supported by the Startup Funds of Great Bay University (No. YJKY250111) and the Innovative Team Program for Regular Universities in Guangdong Province (No. 2025KCXTD031). Qi Wu was supported by the CityU-JD Digits Joint Laboratory in Financial Technology and Engineering, the Hong Kong Research Grants Council General Research Fund (Nos. 11219420/9043008 and 11200219/9042900), the HK Institute of Data Science, the InnoHK initiative of the Government of the HKSAR, and the Laboratory for AI-Powered Financial Technologies. Zhiheng Zhang was supported by the Fundamental Research Funds for the Central Universities (No. 2025110602), the Independent Research Project funded by the School of Statistics and Data Science (No. 2026110081), and the Shanghai Engineering Research Center of Finance Intelligence (No. 19DZ2254600).
	\bibliography{reference.bib}
	\appendix
	
	\clearpage

\section{Wasserstein barycenter representation}\label{app:barycenter}
\begin{proof}[Proof of Proposition~\ref{prop:barycenter_quantile_main}]
	Let $\mathcal{Q}$ be the set of all left-continuous, non-decreasing functions $q:[0,1] \to \mathbb{R}$. 
	In one dimension, the map $\nu \mapsto \nu^{-1}$ is an isometry between $(\mathcal P_2(\mathbb{R}), \mathcal{W}_2)$ and $(\mathcal{Q}, \|\cdot\|_{L_2})$ (see Proposition~\ref{prop:Quantile representation of Wasserstein}).
	
	Fix a policy $\pi$ and denote the random quantile curve as $Q(t) := \mathcal{Y}[\pi(X)]^{-1}(t)$. 
	By Proposition~\ref{prop:Quantile representation of Wasserstein}, for any measure $\nu \in \mathcal P_2(\mathbb{R})$ with quantile function $q_\nu := \nu^{-1} \in \mathcal{Q}$, we have:
	\[
	\mathbb{E}\big[\mathcal{W}_2^2(\mathcal{Y}[\pi(X)], \nu)\big] 
	= \mathbb{E} \left[ \int_0^1 |Q(t) - q_\nu(t)|^2 \, dt \right] 
	= \int_0^1 \mathbb{E} \big[ |Q(t) - q_\nu(t)|^2 \big] \, dt.
	\]
	For each fixed $t \in [0,1]$, we expand the integrand:
	\[
	\mathbb{E}\big[|Q(t) - q_\nu(t)|^2\big] 
	= \mathbb{E}\big[|Q(t) - \mathbb{E}[Q(t)]|^2\big] + |\mathbb{E}[Q(t)] - q_\nu(t)|^2,
	\]
	where the cross term vanishes since $\mathbb{E}[Q(t) - \mathbb{E}[Q(t)]] = 0$. 
	Consequently,
	\[
	\mathbb{E}\big[\mathcal{W}_2^2(\mathcal{Y}[\pi(X)], \nu)\big] 
	= C_\pi + \int_0^1 |\mathbb{E}[Q(t)] - q_\nu(t)|^2 \, dt,
	\]
	where $C_\pi := \int_0^1 \mathbb{E}[|Q(t) - \mathbb{E}[Q(t)]|^2] \, dt$ is a constant independent of $\nu$.
	
	It follows that any minimizer must satisfy $q_\nu(t) = \mathbb{E}[Q(t)]$ for almost every $t$. 
	By Assumption~\ref{ass:boundedness}, $Q$ is uniformly bounded by $M$; hence, the map $t \mapsto \mathbb{E}[Q(t)]$ is well-defined and belongs to $L_2([0,1])$. 
	Furthermore, since $Q(\cdot)$ is non-decreasing and left-continuous almost surely, the pointwise expectation $\mathbb{E}[Q(\cdot)]$ is also non-decreasing. 
	Left-continuity of $\mathbb{E}[Q(\cdot)]$ follows from the dominated convergence: for any $t \in (0,1]$, letting $s \uparrow t$, we have $\mathbb{E}[Q(s)] \to \mathbb{E}[Q(t)]$ because $|Q(s)| \le M$ and $Q(s) \to Q(t)$ a.s.
	
	Thus, $\mathbb{E}[Q] \in \mathcal{Q}$, ensuring that a minimizer exists within $\mathcal{Q}$. 
	Uniqueness follows from the strict convexity of the functional $q \mapsto \int_0^1 |\mathbb{E}[Q(t)] - q(t)|^2 \, dt$ on $\mathcal{Q}$. 
	Therefore, the unique barycenter $\mu(\pi)$ satisfies:
	\[
	(\mu(\pi))^{-1}(t) = \mathbb{E}[\mathcal{Y}[\pi(X)]^{-1}(t)], \qquad \forall t \in [0,1].
	\]
\end{proof}

\section{Proof of Theorem~\ref{thm:ipw_main}}\label{app:proof_ipw}
Before presenting the proof of Theorem \ref{thm:ipw_main}, we first give two lemmas that are useful for our proofs.

\begin{lemma}\label{lemma:sup all bound}
Given \(\delta>0\), \(\mathcal{M}>0\), \(\epsilon\geq 0\), \(L\geq 0\), \(\beta\in(0,1]\), and a probability space \((\Omega,\mathcal{F},\mathbb{P})\). Suppose that \(f(t;\omega):\mathcal{T}\times\Omega\rightarrow\mathbb{R}\) with \(\mathcal{T}=[t_{\mathrm{Low}},t_{\mathrm{Upp}}]\) and
\begin{equation*}
\begin{aligned}
|f(t;\omega)-f(s;\omega)|\leq L|t-s|^{\beta}+\epsilon\quad\forall\;\omega\in\Omega.
\end{aligned}
\end{equation*}
Let \(\mathcal{T}_{\mathrm{par}}:=\{t_{\mathrm{Low}}=t_{0}<t_{1}<\cdots<t_{\mathcal{J}}=t_{\mathrm{Upp}}\}\) be a partition of \([t_{\mathrm{Low}},t_{\mathrm{Upp}}]\) with \(t_{j}=t_{\mathrm{Low}}+\frac{j(t_{\mathrm{Upp}}-t_{\mathrm{Low}})}{\mathcal{J}}\) and
\begin{equation*}
\begin{aligned}
\mathbb{P}\left\{f(t;\omega)\geq \mathcal{M}-\epsilon-L\left|\frac{t_{\mathrm{Upp}}-t_{\mathrm{Low}}}{\mathcal{J}}\right|^{\beta}\right\}\leq \frac{\delta}{\mathcal{J}},
\end{aligned}
\end{equation*}
for any \(t\in\mathcal{T}_{\mathrm{par}}\), then we have
\begin{equation*}
\begin{aligned}
\mathbb{P}\{\underset{t\in\mathcal{T}}{\sup}\;f(t;\omega)\geq \mathcal{M}\}\leq \delta.
\end{aligned}
\end{equation*}
\end{lemma}
\begin{proof}
For any \(\eta>0\), there exists \(\tilde{t}\) such that 
\begin{equation*}
\begin{aligned}
\underset{t\in\mathcal{T}}{\sup}\;f(t;\omega)-\eta<f(\tilde{t};\omega)\leq \underset{t\in\mathcal{T}}{\sup}\;f(t;\omega).
\end{aligned}
\end{equation*}
Now, choose \(t_{j}\in\mathcal{T}_{\mathrm{par}}\) such that \(t_{j-1}<\tilde{t}\leq t_{j}\). Then we have
\begin{equation*}
\begin{aligned}
\underset{t\in\mathcal{T}}{\sup}\;f(t;\omega)&=\underset{t\in\mathcal{T}}{\sup}\;f(t;\omega)-f(\tilde{t};\omega)-\eta+\eta+f(\tilde{t};\omega)\\
&\leq f(\tilde{t};\omega)+\eta=f(\tilde{t};\omega)-f(t_{j};\omega)+f(t_{j};\omega)+\eta\leq |f(\tilde{t};\omega)-f(t_{j};\omega)|+\underset{t\in\mathcal{T}_{\mathrm{par}}}{\max}\;f(t;\omega)+\eta\\
&\leq \epsilon+L|\tilde{t}-t_{j}|^{\beta}+\underset{t\in\mathcal{T}_{\mathrm{par}}}{\max}\;f(t;\omega)+\eta\leq \epsilon+L|t_{j}-t_{j-1}|^{\beta}+\underset{t\in\mathcal{T}_{\mathrm{par}}}{\max}\;f(t;\omega)+\eta\\
&=\epsilon+L\left|\frac{t_{\mathrm{Upp}}-t_{\mathrm{Low}}}{\mathcal{J}}\right|^{\beta}+\underset{t\in\mathcal{T}_{\mathrm{par}}}{\max}\;f(t;\omega)+\eta.
\end{aligned}
\end{equation*}
Since \(\eta\) is arbitrary, we must have \(\underset{t\in\mathcal{T}}{\sup}\;f(t;\omega)\leq \epsilon+L\left|\frac{t_{\mathrm{Upp}}-t_{\mathrm{Low}}}{\mathcal{J}}\right|^{\beta}+\underset{t\in\mathcal{T}_{\mathrm{par}}}{\max}\;f(t;\omega)\). As a result, we have
\begin{align*}
&\mathbb{P}\{\underset{t\in\mathcal{T}}{\sup}\;f(t;\omega)\geq \mathcal{M}\}\leq\mathbb{P}\left\{\underset{t\in\mathcal{T}_{\mathrm{par}}}{\max}\;f(t;\omega)\geq \mathcal{M}-\epsilon-L\left|\frac{t_{\mathrm{Upp}}-t_{\mathrm{Low}}}{\mathcal{J}}\right|^{\beta}\right\}\\
&\leq\underset{j=1}{\overset{\mathcal{J}}{\sum}}\;\mathbb{P}\left\{f(t_{j};\omega)\geq \mathcal{M}-\epsilon-L\left|\frac{t_{\mathrm{Upp}}-t_{\mathrm{Low}}}{\mathcal{J}}\right|^{\beta}\right\}\leq\delta.
\end{align*}
\end{proof}
\begin{lemma}[Multi-class Sauer bound \cite{haussler1995generalization}]\label{lem:sauer-multi-class}
	Let $\Pi$ be a class of functions mapping $\mathcal X$ to a $d$-element set $\mathcal A$.
	Let $V:=\mathrm{N\text{-}dim}(\Pi)$.
	For any sample $x_{1:n}$ with $n\ge V$, the number of distinct labelings
	\[
	m_\Pi(n):=\left|\left\{(\pi(x_1),\dots,\pi(x_n)):\ \pi\in\Pi\right\}\right|
	\]
	satisfies
	\[
	m_\Pi(n)\le (e n d)^V,
	\qquad\text{and hence}\qquad
	\log m_\Pi(n)\le V\log(e n d).
	\]
\end{lemma}

\begin{proof}
	This is a standard consequence of the multi-class Sauer-type lemma for the Natarajan dimension.
	We include the (weak but sufficient) form used in the main text: $m_\Pi(n)$ grows at most polynomially in $n$ with exponent $V$
	and an additional factor $d$ for the label set.
\end{proof}
\noindent Now we will prove Theorem~\ref{thm:ipw_main} in the following.
\begin{proof}[Proof of Theorem~\ref{thm:ipw_main}]
	Recall $q_\pi(t)=\mu(\pi)^{-1}(t)=\mathbb E[\mathcal Y[\pi(X)]^{-1}(t)]$.
	By unconfoundedness and overlap, for each fixed $(\pi,t)$,
	\[
	q_\pi(t)=\mathbb E\!\left[\frac{\mathbf 1_{\{A=\pi(X)\}}\mathcal Y^{-1}(t)}{f_0(A|X)}\right].
	\]
	Also recall $\hat q_\pi^{\mathrm{IPW}}(t)=\mathbb P_N\left[\mathbf 1_{\{A=\pi(X)\}}\mathcal Y^{-1}(t)/f_0(A|X)\right]$.
	
	First, the regret can be reduced to uniform estimation error.
	Let $\hat\pi^{\mathrm{IPW}}\in\arg\max_{\pi\in\Pi}U(\hat\mu^{\mathrm{IPW}}(\pi))$.
	By the argmax property and Lipschitzness of $U$, we have
	\begin{equation*}
		\begin{aligned}
			\mathcal{R}(\hat{\pi}^{\mathrm{IPW}}) 
			&= U(\mu(\pi^\star)) - U(\mu(\hat{\pi}^{\mathrm{IPW}})) \\
			&= \Big( U(\mu(\pi^\star)) - U(\hat{\mu}^{\mathrm{IPW}}(\pi^\star)) \Big) + \Big( U(\hat{\mu}^{\mathrm{IPW}}(\pi^\star)) - U(\hat{\mu}^{\mathrm{IPW}}(\hat{\pi}^{\mathrm{IPW}})) \Big) \\
			&\quad + \Big( U(\hat{\mu}^{\mathrm{IPW}}(\hat{\pi}^{\mathrm{IPW}})) - U(\mu(\hat{\pi}^{\mathrm{IPW}})) \Big) \\
			&\le \Big( U(\mu(\pi^\star)) - U(\hat{\mu}^{\mathrm{IPW}}(\pi^\star)) \Big) + 0 + \Big( U(\hat{\mu}^{\mathrm{IPW}}(\hat{\pi}^{\mathrm{IPW}})) - U(\mu(\hat{\pi}^{\mathrm{IPW}})) \Big) \\
			&\le \sup_{\pi \in \Pi} \big| U(\mu(\pi)) - U(\hat{\mu}^{\mathrm{IPW}}(\pi)) \big| + \sup_{\pi \in \Pi} \big| U(\hat{\mu}^{\mathrm{IPW}}(\pi)) - U(\mu(\pi)) \big| \\
			&= 2 \sup_{\pi \in \Pi} \big| U(\hat{\mu}^{\mathrm{IPW}}(\pi)) - U(\mu(\pi)) \big| \\
			&\le 2 L_U \sup_{\pi \in \Pi} \mathcal{W}_2(\hat{\mu}^{\mathrm{IPW}}(\pi), \mu(\pi)),
		\end{aligned}
	\end{equation*}
	To handle the monotonicity constraint, we use the fact that $q_\pi \in \mathcal Q$. Since $\Pi_{\mathcal Q}$ is a non-expansive contraction in $L_2$, we have:
	\[
	\mathcal W_2(\hat\mu^{\mathrm{IPW}}(\pi),\mu(\pi)) = \|\hat q_\pi^{\mathrm{IPW},\uparrow} - q_\pi\|_{L_2} \le \|\hat q_\pi^{\mathrm{IPW}} - q_\pi\|_{L_2} \le \sup_{t\in[0,1]}\big|\hat q_\pi^{\mathrm{IPW}}(t)-q_\pi(t)\big|.
	\]
	Therefore, it follows that
	\[
	\mathcal R(\hat\pi^{\mathrm{IPW}}) \le 2L_U\sup_{t\in[0,1]}\sup_{\pi\in\Pi}\big|\hat q_\pi^{\mathrm{IPW}}(t)-q_\pi(t)\big|.
	\]
	Define
	\[
	\Delta(t):=\sup_{\pi\in\Pi}\left|(\mathbb P_N-\mathbb P)\left[\frac{\mathbf 1_{\{A=\pi(X)\}}\mathcal Y^{-1}(t)}{f_0(A|X)}\right]\right|.
	\]
	It suffices to bound $\sup_{t\in[0,1]}\Delta(t)$.
	
	For any fixed $t$, conditional on $X_{1:N}$, the set of distinct labelings
	$\{(\pi(X_1),\dots,\pi(X_N)):\pi\in\Pi\}$ has cardinality $m_\Pi(N)$.
	For a fixed labeling, the summands are i.i.d.\ and bounded by
	$|\mathcal Y^{-1}(t)|/f_0(A|X)\le M/\underline f$.
	Thus Hoeffding's inequality and a union bound over labelings yield
	\[
	\mathbb P\{\Delta(t)\ge \varepsilon\}\le 2\,m_\Pi(N)\exp\!\left(-\frac{N\varepsilon^2}{2(M/\underline f)^2}\right).
	\]
	
	Let $\mathcal T_{\mathrm{par}}=\{t_j:j=0,\dots,\mathcal J\}$ be the uniform grid with $\mathcal J+1$ points (mesh $\eta=1/\mathcal J$).
	A union bound over $t_j$ yields: with probability at least $1-\delta$,
	\[
	\max_{0\le j\le \mathcal J}\Delta(t_j)
	\le
	\frac{M}{\underline f}\sqrt{\frac{2\log\left(\frac{2(\mathcal J+1)m_\Pi(N)}{\delta}\right)}{N}}.
	\]
	
	Now we will extend from grid to continuum by Lipschitz regularity.
	By Assumption~\ref{ass:quantile_lipschitz}, $\Delta(\cdot)$ is pathwise Lipschitz with constant $\mathcal U/\underline f$ for both $\mathbb{P}$ and $\mathbb{P}_N$ terms. For any $t$, let $t_j$ be the nearest grid point such that $|t-t_j| \le \eta$. Then:
	\[
	\Delta(t) \le \Delta(t_j) + \sup_{\pi}\left| \left(\frac{\mathcal{Y}^{-1}(t)}{f_0} - \frac{\mathcal{Y}^{-1}(t_j)}{f_0}\right) - \mathbb{E}\left[\frac{\mathcal{Y}^{-1}(t)}{f_0} - \frac{\mathcal{Y}^{-1}(t_j)}{f_0}\right] \right| \le \Delta(t_j) + \frac{2\mathcal{U}}{\underline{f}}\eta.
	\]
	Taking the supremum over $t \in [0,1]$, we obtain $\sup_t \Delta(t) \le \max_j \Delta(t_j) + (2\mathcal U/\underline f)\eta$.
	
	Finally, we convert $m_\Pi(N)$ to Natarajan dimension. Recall $V:=\mathrm{N\text{-}dim}(\Pi)$.
	By Lemma~\ref{lem:sauer-multi-class}, $\log m_\Pi(N)\le V\log(eNd)$.
	Using $\sqrt{a+b}\le \sqrt a+\sqrt b$, we obtain
	\[
	\max_{j}\Delta(t_j)
	\le
	\frac{M}{\underline f}\left(
	\sqrt{\frac{2V\log(eNd)}{N}}
	+
	\sqrt{\frac{2\log\left(\frac{2(\mathcal J+1)}{\delta}\right)}{N}}\right).
	\]
	Combining the above bounds yields the theorem statement.
\end{proof}

\section{Proof of Theorem~\ref{thm:statistical guarantee DR}}\label{app:proof_DR}
\begin{proof}[Proof of Theorem \ref{thm:statistical guarantee DR}]

	Recall $q_\pi(t)=\mu(\pi)^{-1}(t)=\mathbb E[\mathcal Y[\pi(X)]^{-1}(t)]$.
	By unconfoundedness,
	\[
	q_\pi(t)=\mathbb E[m_0(\pi(X),X)(t)],
	\qquad
	m_0(a,x)(t):=\mathbb E[\mathcal Y^{-1}(t)\mid A=a,X=x].
	\]
	For notational convenience, define the reciprocal weights by
	\[
	g_0(a,x):=\frac{1}{f_0(a|x)},
	\qquad
	g(a,x):=g_0(a,x),
	\qquad
	\hat g_0^\ell(a,x):=\frac{1}{\hat f_0^\ell(a|x)}.
	\]
	When $\hat f_0^\ell$ appears in a denominator, it is maintained under the bounded-away-from-zero convention $\hat f_0^\ell(a|x)\ge\underline f$.
	Throughout the proof, we condition on the training folds $\mathcal I_{-\ell}$ used to construct
	$\hat f_0^\ell$ and $\hat m_0^\ell$; hence the nuisance estimates are fixed when taking expectations over the held-out fold $\mathcal I_\ell$.

	As in scalar policy learning, the regret can be decomposed as
\begin{equation*}
	\begin{aligned}
		\mathcal{R}(\hat{\pi}^{\mathrm{DR}}) 
		&= U(\mu(\pi^\star)) - U(\mu(\hat{\pi}^{\mathrm{DR}})) \\
		&= \Big( U(\mu(\pi^\star)) - U(\hat{\mu}^{\mathrm{DR}}(\pi^\star)) \Big) + \Big( U(\hat{\mu}^{\mathrm{DR}}(\pi^\star)) - U(\hat{\mu}^{\mathrm{DR}}(\hat{\pi}^{\mathrm{DR}})) \Big) \\
		&\quad + \Big( U(\hat{\mu}^{\mathrm{DR}}(\hat{\pi}^{\mathrm{DR}})) - U(\mu(\hat{\pi}^{\mathrm{DR}})) \Big) \\
		&\le \Big( U(\mu(\pi^\star)) - U(\hat{\mu}^{\mathrm{DR}}(\pi^\star)) \Big) + 0 + \Big( U(\hat{\mu}^{\mathrm{DR}}(\hat{\pi}^{\mathrm{DR}})) - U(\mu(\hat{\pi}^{\mathrm{DR}})) \Big) \\
		&\le \sup_{\pi \in \Pi} \big| U(\mu(\pi)) - U(\hat{\mu}^{\mathrm{DR}}(\pi)) \big| + \sup_{\pi \in \Pi} \big| U(\hat{\mu}^{\mathrm{DR}}(\pi)) - U(\mu(\pi)) \big| \\
		&= 2 \sup_{\pi \in \Pi} \big| U(\hat{\mu}^{\mathrm{DR}}(\pi)) - U(\mu(\pi)) \big| \\
		&\le 2 L_U \sup_{\pi \in \Pi} \mathcal{W}_2(\hat{\mu}^{\mathrm{DR}}(\pi), \mu(\pi)).
	\end{aligned}
\end{equation*}

Using the $L_2$ representation of the Wasserstein distance and the non-expansiveness of the isotonic projection $\Pi_{\mathcal{Q}}$, we observe that for any fixed $\pi \in \Pi$:
\begin{equation*}
	\mathcal{W}_2(\hat{\mu}^{\mathrm{DR}}(\pi), \mu(\pi)) = \|\hat{q}_\pi^{\mathrm{DR},\uparrow} - q_\pi\|_{L_2} \le \|\hat{q}_\pi^{\mathrm{DR}} - q_\pi\|_{L_2} = \left( \int_0^1 \big| \hat{q}_\pi^{\mathrm{DR}}(t) - q_\pi(t) \big|^2 dt \right)^{1/2}.
\end{equation*}
To obtain a uniform bound over $\Pi$, we note that for each $t \in [0,1]$ and any $\pi \in \Pi$:
\begin{equation*}
	\big| \hat{q}_\pi^{\mathrm{DR}}(t) - q_\pi(t) \big| \le \sup_{\pi' \in \Pi} \big| \hat{q}_{\pi'}^{\mathrm{DR}}(t) - q_{\pi'}(t) \big|.
\end{equation*}
Squaring both sides and integrating over $t \in [0,1]$, we obtain:
\begin{equation*}
	\int_0^1 \big| \hat{q}_\pi^{\mathrm{DR}}(t) - q_\pi(t) \big|^2 dt \le \int_0^1 \left[ \sup_{\pi' \in \Pi} \big| \hat{q}_{\pi'}^{\mathrm{DR}}(t) - q_{\pi'}(t) \big| \right]^2 dt.
\end{equation*}
Since the right-hand side is independent of $\pi$, we can take the supremum over $\pi \in \Pi$ on the left-hand side and take the square root:
\begin{equation*}
	\sup_{\pi \in \Pi} \left( \int_0^1 \big| \hat{q}_\pi^{\mathrm{DR}}(t) - q_\pi(t) \big|^2 dt \right)^{1/2} \le \left( \int_0^1 \left[ \sup_{\pi \in \Pi} \big| \hat{q}_\pi^{\mathrm{DR}}(t) - q_\pi(t) \big| \right]^2 dt \right)^{1/2}.
\end{equation*}
Consequently, define
\[
D(t):=\sup_{\pi\in\Pi}\big|\hat q_\pi^{\mathrm{DR}}(t)-q_\pi(t)\big|.
\]
Then the preceding display gives
\begin{equation*}
	\mathcal{R}(\hat{\pi}^{\mathrm{DR}}) \le 2 L_U \left( \int_0^1 D(t)^2\,dt \right)^{1/2}.
\end{equation*}

		Therefore, now it suffices to bound \(D(t)\), which can be split into three main steps:
		\paragraph{Step I: Decomposing \(D(t)\).}
		Note that
        {\small
			\begin{align*}
				&D(t)\\
				&=\underset{\pi\in\Pi}{\sup}\;\left|q_\pi(t)-\hat q_\pi^{\mathrm{DR}}(t)\right|\\
				&=\underset{\pi\in\Pi}{\sup}\;\left|
				\mathbb{E}\left[\mathcal{Y}[\pi(X)]^{-1}(t)\right]-\frac{1}{L}\underset{\ell=1}{\overset{L}{\sum}}\;\mathbb{P}_{n,\ell}\left[\hat{m}_{0}^{\ell}(\pi(X),X)+\frac{\mathbf{1}_{\{A=\pi(X)\}}}{\hat f_{0}^{\ell}(A|X)}(\mathcal{Y}^{-1}-\hat{m}_{0}^{\ell}(\pi(X),X))\right](t)
				\right|.
			\end{align*}
        }\noindent

		Denote
	\(\hat{g}_{0}^{\ell}(\cdot,\cdot)=\frac{1}{\hat f_{0}^{\ell}(\cdot|\cdot)}\). We then have
		{\small
			\begin{align}
				&D(t)\nonumber\\
				&=\underset{\pi\in\Pi}{\sup}\;\left|
				\mathbb{E}\left[\mathcal{Y}[\pi(X)]^{-1}(t)\right]-\frac{1}{L}\underset{\ell=1}{\overset{L}{\sum}}\;\mathbb{P}_{n,\ell}\left[\hat{m}_{0}^{\ell}(\pi(X),X)+\frac{\mathbf{1}_{\{A=\pi(X)\}}}{\hat f_{0}^{\ell}(A|X)}(\mathcal{Y}^{-1}-\hat{m}_{0}^{\ell}(\pi(X),X))\right](t)
				\right|\nonumber\\
				&\leq \underset{\pi\in\Pi}{\sup}\left|
				\frac{1}{L}\;\underset{\ell=1}{\overset{L}{\sum}}\;\mathbb{P}_{n,\ell}\left\{
				\begin{aligned}
					&\hat{m}_{0}^{\ell}(\pi(X),X)+\frac{\mathbf{1}_{\{A=\pi(X)\}}}{\hat f_{0}^{\ell}(A|X)}(\mathcal{Y}^{-1}-\hat{m}_{0}^{\ell}(\pi(X),X))\\
					&-m_{0}(\pi(X),X)-\frac{\mathbf{1}_{\{A=\pi(X)\}}}{f_{0}(A|X)}(\mathcal{Y}^{-1}-m_{0}(\pi(X),X))
				\end{aligned}
				\right\}(t)
				\right|\nonumber\\
				&\quad+\underset{\pi\in\Pi}{\sup}\left|\frac{1}{L}\;\underset{\ell=1}{\overset{L}{\sum}}\;\mathbb{P}_{n,\ell}\left\{m_{0}(\pi(X),X)+\frac{\mathbf{1}_{\{A=\pi(X)\}}}{f_{0}(A|X)}(\mathcal{Y}^{-1}-m_{0}(\pi(X),X))\right\}(t)-\mathbb{E}[\mathcal{Y}^{-1}[\pi(X)]](t)\right|\nonumber\\
				&\leq \underset{\pi\in\Pi}{\sup}\left|
				\frac{1}{L}\;\underset{\ell=1}{\overset{L}{\sum}}\;\mathbb{P}_{n,\ell}\mathbf{1}_{\{A=\pi(X)\}}\bigg(\mathcal{Y}^{-1}-\hat{m}_{0}^{\ell}(\pi(X),X)\bigg)(t)\bigg(\hat{g}_{0}^{\ell}(A,X)-g(A,X)\bigg)
				\right|\label{eqt:final result 2 no expectation-DR}\\
				&\quad+\underset{\pi\in\Pi}{\sup}\left|
				\frac{1}{L}\;\underset{\ell=1}{\overset{L}{\sum}}\;\mathbb{P}_{n,\ell}\left(\mathbf{1}_{\{A=\pi(X)\}}g(A,X)-1\right)\bigg(m_{0}(\pi(X),X)-\hat{m}_{0}^{\ell}(\pi(X),X)\bigg)(t)
				\right|\label{eqt:final result 3 no expectation-DR}\\
				&\quad+\underset{\pi\in\Pi}{\sup}\left|\mathbb{P}_{N}\left\{m_{0}(\pi(X),X)+\frac{\mathbf{1}_{\{A=\pi(X)\}}}{f_{0}(A|X)}(\mathcal{Y}^{-1}-m_{0}(\pi(X),X))\right\}(t)-\mathbb{E}[\mathcal{Y}^{-1}[\pi(X)]](t)\right|.\label{eqt:final result 4 no expectation-DR}
			\end{align}
		}\noindent

		The goal is finding stochastic bounds of Eqns. \eqref{eqt:final result 2 no expectation-DR} - \eqref{eqt:final result 4 no expectation-DR}.

		Before bounding these terms, we record the two orthogonality identities that remove the standalone first-order nuisance bias.
		For any fixed $\pi$, $t$, and fold $\ell$, conditioning on $\mathcal I_{-\ell}$ gives
		\begin{align}
		&\mathbb E\!\left[
		\mathbf 1_{\{A=\pi(X)\}}
		\big(\mathcal Y^{-1}(t)-m_0(\pi(X),X)(t)\big)
		\big(\hat g_0^\ell(A,X)-g(A,X)\big)
		\mid \mathcal I_{-\ell}
		\right]
		=0,
		\label{eq:orth_propensity_residual}
		\\
		&\mathbb E\!\left[
		\left(\mathbf 1_{\{A=\pi(X)\}}g(A,X)-1\right)
		\big(m_0(\pi(X),X)-\hat m_0^\ell(\pi(X),X)\big)(t)
		\mid \mathcal I_{-\ell}
		\right]
		=0.
		\label{eq:orth_outcome_residual}
		\end{align}
		Indeed, \eqref{eq:orth_propensity_residual} follows from
		\[
		\mathbb E\!\left[
		\mathcal Y^{-1}(t)-m_0(A,X)(t)
		\mid A,X
		\right]=0,
		\]
		while \eqref{eq:orth_outcome_residual} follows from
		\[
		\mathbb E\!\left[
		\mathbf 1_{\{A=\pi(X)\}}g(A,X)-1
		\mid X
		\right]
		=
		f_0(\pi(X)|X)\frac{1}{f_0(\pi(X)|X)}-1
		=0.
		\]
		Thus, the only non-centered nuisance bias is of product order in the propensity and outcome-regression errors.

		\paragraph{Step II: Studying Eqns. \eqref{eqt:final result 2 no expectation-DR} and \eqref{eqt:final result 3 no expectation-DR}.}
		We now bound the two nuisance-dependent terms in the corrected decomposition.
		
		We first bound Eqn. \eqref{eqt:final result 2 no expectation-DR}. Indeed, we have
		{\small
			\begin{align}
				&\text{Eqn. \eqref{eqt:final result 2 no expectation-DR}}\nonumber\\
				&=\underset{\pi\in\Pi}{\sup}\left|
				\frac{1}{L}\;\underset{\ell=1}{\overset{L}{\sum}}\;\mathbb{P}_{n,\ell}\mathbf{1}_{\{A=\pi(X)\}}\bigg(\mathcal{Y}^{-1}-\hat{m}_{0}^{\ell}(\pi(X),X)\bigg)(t)\bigg(\hat{g}_{0}^{\ell}(A,X)-g(A,X)\bigg)
				\right|\nonumber\\
				&\leq \frac{1}{L}\;\underset{\ell=1}{\overset{L}{\sum}}\;\underset{\pi\in\Pi}{\sup}\left|\mathbb{P}_{n,\ell}\left(
				\begin{aligned}
					&\mathbf{1}_{\{A=\pi(X)\}}\bigg(\mathcal{Y}^{-1}-\hat{m}_{0}^{\ell}(\pi(X),X)\bigg)(t)\bigg(\hat{g}_{0}^{\ell}(A,X)-g(A,X)\bigg)\\
					&-\mathbb{E}\left[\mathbf{1}_{\{A=\pi(X)\}}\bigg(\mathcal{Y}^{-1}-\hat{m}_{0}^{\ell}(\pi(X),X)\bigg)(t)\bigg(\hat{g}_{0}^{\ell}(A,X)-g(A,X)\bigg)\right]
				\end{aligned}
				\right)
				\right|\label{eqt:bound 2 DR 1}\\
				&\quad+\frac{1}{L}\;\underset{\ell=1}{\overset{L}{\sum}}\;\underset{\pi\in\Pi}{\sup}\left|
				\mathbb{E}\left[\mathbf{1}_{\{A=\pi(X)\}}\bigg(\mathcal{Y}^{-1}-\hat{m}_{0}^{\ell}(\pi(X),X)\bigg)(t)\bigg(\hat{g}_{0}^{\ell}(A,X)-g(A,X)\bigg)\right]\right|.\label{eqt:bound 2 DR 2}
			\end{align} 
		}\noindent
		We bound Eqn. \eqref{eqt:bound 2 DR 1}. Denote
		{\small
			\begin{equation*}
				\begin{aligned}
					\Gamma_{2}^{\ell}(t):=\underset{\pi\in\Pi}{\sup}\;\left|\mathbb{P}_{n,\ell}\left(
					\begin{aligned}
						&\mathbf{1}_{\{A=\pi(X)\}}\bigg(\mathcal{Y}^{-1}-\hat{m}_{0}^{\ell}(\pi(X),X)\bigg)(t)\bigg(\hat{g}_{0}^{\ell}(A,X)-g(A,X)\bigg)\\
						&-\mathbb{E}\left[\mathbf{1}_{\{A=\pi(X)\}}\bigg(\mathcal{Y}^{-1}-\hat{m}_{0}^{\ell}(\pi(X),X)\bigg)(t)\bigg(\hat{g}_{0}^{\ell}(A,X)-g(A,X)\bigg)\right]
					\end{aligned}
					\right)
					\right|.
				\end{aligned}
			\end{equation*}
		}\noindent
		Also, denote
		\begin{equation*}
			\begin{aligned}
				&\left|\mathbb{P}_{n,\ell}\left(
				\begin{aligned}
					&\mathbf{1}_{\{A=a_{i}\}}\bigg(\mathcal{Y}^{-1}-\hat{m}_{0}^{\ell}(a_{i},X)\bigg)(t)\bigg(\hat{g}_{0}^{\ell}(A,X)-g(A,X)\bigg)\\
					&-\mathbb{E}\left[\mathbf{1}_{\{A=a_{i}\}}\bigg(\mathcal{Y}^{-1}-\hat{m}_{0}^{\ell}(a_{i},X)\bigg)(t)\bigg(\hat{g}_{0}^{\ell}(A,X)-g(A,X)\bigg)\right]
				\end{aligned}
				\right)\right|:=\Gamma_{2}^{\ell;i}(t).
			\end{aligned}
		\end{equation*}
		Define $n_\ell := |\mathcal I_\ell| = N/L$. Using Hoeffding's inequality, we have for any $\varepsilon>0$,
		\begin{equation*}
			\mathbb{P}\left\{\Gamma_{2}^{\ell;i}(t)\geq \varepsilon\big|\mathcal{I}_{-\ell}\right\}\leq 2\exp\left(-\frac{n_\ell\varepsilon^2}{2B^2}\right),
		\end{equation*}
		where 
		\begin{equation*}
			\begin{aligned}
				&\left|\mathbf{1}_{\{A=a_{i}\}}\bigg(\mathcal{Y}^{-1}-\hat{m}_{0}^{\ell}(a_{i},X)\bigg)(t)\bigg(\hat{g}_{0}^{\ell}(A,X)-g(A,X)\bigg)\right|\\
				&\leq\left|\mathbf{1}_{\{A=a_{i}\}}\bigg(\mathcal{Y}^{-1}-m_{0}(a_{i},X)\bigg)(t)\bigg(\hat{g}_{0}^{\ell}(A,X)-g(A,X)\bigg)\right|\\
				&\quad+\left|\mathbf{1}_{\{A=a_{i}\}}\bigg(m_{0}(a_{i},X)-\hat{m}_{0}^{\ell}(a_{i},X)\bigg)(t)\bigg(\hat{g}_{0}^{\ell}(A,X)-g(A,X)\bigg)\right|\\
				&\leq \frac{\|\hat{f}_{0}^{\ell}-f_{0}\|_{\infty}\;\bigg(2M+\|\hat{m}_{0}^{\ell}-m_{0}\|_{\infty,[0,1]}\bigg)}{\underline{f}^{2}}:=B.
			\end{aligned}
		\end{equation*}
		Since there are at most $m_{\Pi}(n_\ell)$ distinct labeling patterns realized by $\Pi$ on $\mathcal{I}_{\ell}$. Applying a union bound over these patterns yields
		\begin{equation*}
			\mathbb{P}\left\{\Gamma_{2}^{\ell}(t)\geq\varepsilon\big|\mathcal{I}_{-\ell}\right\}\leq 2m_{\Pi}(n_\ell)\exp\left(-\frac{n_\ell\varepsilon^2}{2B^2}\right).
		\end{equation*}
		Choose $\varepsilon$ so that the RHS equals $\delta$:
		\begin{equation*}
			2\,m_{\Pi}(n_\ell)\exp\!\left(-\frac{n_\ell\varepsilon^2}{2B^{2}}\right)=\delta\quad\Longleftrightarrow\quad\varepsilon=B\sqrt{\frac{2\log\left(\frac{2m_{\Pi}(n_\ell)}{\delta}\right)}{n_\ell}}.
		\end{equation*}
		As a result, for each fixed \(t\in[0,1]\), with probability at least \(1-\delta\), 
		{\small
			\begin{equation*}
				\begin{aligned}
					&\Gamma_{2}^{\ell}(t)\leq \frac{\|\hat{f}_{0}^{\ell}-f_{0}\|_{\infty}\;\bigg(2M+\|\hat{m}_{0}^{\ell}-m_{0}\|_{\infty,[0,1]}\bigg)}{\underline{f}^{2}}\sqrt{\frac{2\log\left(\frac{2m_{\Pi}(n_\ell)}{\delta}\right)}{n_\ell}}\\
					&\leq \frac{2\|\hat{f}_{0}^{\ell}-f_{0}\|_{\infty}\;\bigg(2M+\|\hat{m}_{0}^{\ell}-m_{0}\|_{\infty,[0,1]}\bigg)}{\underline{f}^{2}}\left(\sqrt{\frac{2L\operatorname{N\text{-}dim}(\Pi)\log\left(e\frac{N}{L}d\right)}{N}}+\sqrt{\frac{2L\log\left(\frac{2}{\delta}\right)}{N}}\right).
				\end{aligned}
			\end{equation*}
		}\noindent
		In particular, for each fixed grid point \(t_j\), with probability at least \(1-\frac{\delta}{\mathcal{J}}\),
		{\small
			\begin{equation*}
				\begin{aligned}
					\Gamma_{2}^{\ell}(t_{j})\leq&\frac{2\|\hat{f}_{0}^{\ell}-f_{0}\|_{\infty}\;\bigg(2M+\|\hat{m}_{0}^{\ell}-m_{0}\|_{\infty,[0,1]}\bigg)}{\underline{f}^{2}}\\
                    &\times\left(\sqrt{\frac{2L\operatorname{N\text{-}dim}(\Pi)\log\left(e\frac{N}{L}d\right)}{N}}+\sqrt{\frac{2L\log\left(\frac{2\mathcal{J}}{\delta}\right)}{N}}\right).
				\end{aligned}
			\end{equation*}
		}\noindent
		Next, we consider \(\left|\bigg(\mathcal{Y}^{-1}-\hat{m}_{0}^{\ell}(\pi(X),X)\bigg)(t)-\bigg(\mathcal{Y}^{-1}-\hat{m}_{0}^{\ell}(\pi(X),X)\bigg)(s)\right|\). Note that for any \(s,\;t\),
		\begin{equation*}
			\begin{aligned}
				&\left|\bigg(\mathcal{Y}^{-1}-\hat{m}_{0}^{\ell}(\pi(X),X)\bigg)(t)-\bigg(\mathcal{Y}^{-1}-\hat{m}_{0}^{\ell}(\pi(X),X)\bigg)(s)\right|\\
				&\leq\left|\mathcal{Y}^{-1}(t)-\mathcal{Y}^{-1}(s)\right|+\left|\hat{m}_{0}^{\ell}(\pi(X),X)(s)-\hat{m}_{0}^{\ell}(\pi(X),X)(t)\right|\\
				&\leq\mathcal{U}|t-s|+\left|\hat{m}_{0}^{\ell}(\pi(X),X)(s)-m_{0}(\pi(X),X)(s)\right|\\
				&\quad+\left|m_{0}(\pi(X),X)(s)-m_{0}(\pi(X),X)(t)\right|+\left|m_{0}(\pi(X),X)(t)-\hat{m}_{0}^{\ell}(\pi(X),X)(t)\right|\\
				&\leq\mathcal{U}|t-s|+\|\hat{m}_{0}^{\ell}-m_{0}\|_{\infty,[0,1]}+\mathcal{U}|t-s|+\|\hat{m}_{0}^{\ell}-m_{0}\|_{\infty,[0,1]}\\
				&= 2\mathcal{U}|t-s|+2\;\|\hat{m}_{0}^{\ell}-m_{0}\|_{\infty,[0,1]},
			\end{aligned}
		\end{equation*}
		we thus have
		\begin{align*}
			&\Gamma_{2}^{\ell}(t)-\Gamma_{2}^{\ell}(s)\\
			&\leq \underset{\pi\in\Pi}{\sup}\;\left|
			\begin{aligned}
				&\mathbb{E}\left[\mathbf{1}_{\{A=\pi(X)\}}\bigg(\mathcal{Y}^{-1}-\hat{m}_{0}^{\ell}(\pi(X),X)\bigg)(t)\bigg(\hat{g}_{0}^{\ell}(A,X)-g(A,X)\bigg)\right]\\
				&-\mathbb{E}\left[\mathbf{1}_{\{A=\pi(X)\}}\bigg(\mathcal{Y}^{-1}-\hat{m}_{0}^{\ell}(\pi(X),X)\bigg)(s)\bigg(\hat{g}_{0}^{\ell}(A,X)-g(A,X)\bigg)\right]
			\end{aligned}
			\right|\\
			&\quad+\underset{\pi\in\Pi}{\sup}\;\left|
			\begin{aligned}
				&\mathbb{P}_{n,\ell}\mathbf{1}_{\{A=\pi(X)\}}\bigg(\mathcal{Y}^{-1}-\hat{m}_{0}^{\ell}(\pi(X),X)\bigg)(s)\bigg(\hat{g}_{0}^{\ell}(A,X)-g(A,X)\bigg)\\
				&-\mathbb{P}_{n,\ell}\mathbf{1}_{\{A=\pi(X)\}}\bigg(\mathcal{Y}^{-1}-\hat{m}_{0}^{\ell}(\pi(X),X)\bigg)(t)\bigg(\hat{g}_{0}^{\ell}(A,X)-g(A,X)\bigg)
			\end{aligned}
			\right|\\
			&\leq \frac{4\|\hat{f}_{0}^{\ell}-f_{0}\|_{\infty}\left(\mathcal{U}|t-s|+\|\hat{m}_{0}^{\ell}-m_{0}\|_{\infty,[0,1]}\right)}{\underline{f}^{2}}.
		\end{align*}
		Similarly, we also have \(\Gamma_{2}^{\ell}(s)-\Gamma_{2}^{\ell}(t)\leq \frac{4\|\hat{f}_{0}^{\ell}-f_{0}\|_{\infty}\left(\mathcal{U}|t-s|+\|\hat{m}_{0}^{\ell}-m_{0}\|_{\infty,[0,1]}\right)}{\underline{f}^{2}}\). Combining the two cases, we have \(\left|\Gamma_{2}^{\ell}(t)-\Gamma_{2}^{\ell}(s)\right|\leq \frac{4\|\hat{f}_{0}^{\ell}-f_{0}\|_{\infty}\left(\mathcal{U}|t-s|+\|\hat{m}_{0}^{\ell}-m_{0}\|_{\infty,[0,1]}\right)}{\underline{f}^{2}}\). By Lemma \ref{lemma:sup all bound}, we then have, with probability at least \(1-\delta\),
		{\small
			\begin{equation*}
				\begin{aligned}
					\underset{t\in[0,1]}{\sup}\;\Gamma_{2}^{\ell}(t)\leq&\frac{2\|\hat{f}_{0}^{\ell}-f_{0}\|_{\infty}\;\bigg(2M+\|\hat{m}_{0}^{\ell}-m_{0}\|_{\infty,[0,1]}\bigg)}{\underline{f}^{2}}\\
                    &\times\left(\sqrt{\frac{2L\operatorname{N\text{-}dim}(\Pi)\log\left(e\frac{N}{L}d\right)}{N}}+\sqrt{\frac{2L\log\left(\frac{2\mathcal{J}}{\delta}\right)}{N}}\right)\\
					&+\frac{4\|\hat{f}_{0}^{\ell}-f_{0}\|_{\infty}\;\|\hat{m}_{0}^{\ell}-m_{0}\|_{\infty,[0,1]}}{\underline{f}^{2}}+\frac{4\eta\;\mathcal{U}\|\hat{f}_{0}^{\ell}-f_{0}\|_{\infty}}{\underline{f}^{2}}.
				\end{aligned}
			\end{equation*}
		}\noindent
		As a result, we have, with probability at least \(1-3\delta\),
		{\small
			\begin{equation*}
				\begin{aligned}
					\underset{t\in[0,1]}{\sup}\;\Gamma_{2}^{\ell}(t)\leq&\frac{2\operatorname{Rate}_{f_{0}}(N,\delta)\;\bigg(2M+\operatorname{Rate}_{m_{0}}(N,\delta)\bigg)}{\underline{f}^{2}}\\
                    &\times\left(\sqrt{\frac{2L\operatorname{N\text{-}dim}(\Pi)\log\left(e\frac{N}{L}d\right)}{N}}+\sqrt{\frac{2L\log\left(\frac{2\mathcal{J}}{\delta}\right)}{N}}\right)\\
					&+\frac{4\operatorname{Rate}_{f_{0}}(N,\delta)\;\operatorname{Rate}_{m_{0}}(N,\delta)}{\underline{f}^{2}}+\frac{4\eta\;\mathcal{U}\operatorname{Rate}_{f_{0}}(N,\delta)}{\underline{f}^{2}}.
				\end{aligned}
			\end{equation*}
		}\noindent
		Equivalently, we have, with probability at least \(1-\frac{\delta}{L}\),
		{\small
			\begin{equation*}
				\begin{aligned}
					\underset{t\in[0,1]}{\sup}\;\Gamma_{2}^{\ell}(t)\leq&\frac{2\operatorname{Rate}_{f_{0}}\left(N,\frac{\delta}{3L}\right)\;\bigg(2M+\operatorname{Rate}_{m_{0}}\left(N,\frac{\delta}{3L}\right)\bigg)}{\underline{f}^{2}}\\
                    &\times\left(\sqrt{\frac{2L\operatorname{N\text{-}dim}(\Pi)\log\left(e\frac{N}{L}d\right)}{N}}+\sqrt{\frac{2L\log\left(\frac{6\mathcal{J}L}{\delta}\right)}{N}}\right)\\
					&+\frac{4\operatorname{Rate}_{f_{0}}\left(N,\frac{\delta}{3L}\right)\;\operatorname{Rate}_{m_{0}}\left(N,\frac{\delta}{3L}\right)}{\underline{f}^{2}}+\frac{4\eta\;\mathcal{U}\operatorname{Rate}_{f_{0}}\left(N,\frac{\delta}{3L}\right)}{\underline{f}^{2}},
				\end{aligned}
			\end{equation*}
		}\noindent
		which implies that, with probability at least \(1-\delta\),
		{\small
			\begin{equation}
				\begin{aligned}\label{eqt:temp 2 final result of stochastic bound DR-1}
					\text{Eqn. \eqref{eqt:bound 2 DR 1}}=&\frac{1}{L}\underset{\ell=1}{\overset{L}{\sum}}\;\Gamma_{2}^{\ell}(t)\leq \frac{1}{L}\underset{\ell=1}{\overset{L}{\sum}}\;\underset{t\in[0,1]}{\sup}\;\Gamma_{2}^{\ell}(t)\\
					\leq&\frac{2\operatorname{Rate}_{f_{0}}\left(N,\frac{\delta}{3L}\right)\;\bigg(2M+\operatorname{Rate}_{m_{0}}\left(N,\frac{\delta}{3L}\right)\bigg)}{\underline{f}^{2}}\\
                    &\times\left(\sqrt{\frac{2L\operatorname{N\text{-}dim}(\Pi)\log\left(e\frac{N}{L}d\right)}{N}}+\sqrt{\frac{2L\log\left(\frac{6\mathcal{J}L}{\delta}\right)}{N}}\right)\\
					&+\frac{4\operatorname{Rate}_{f_{0}}\left(N,\frac{\delta}{3L}\right)\;\operatorname{Rate}_{m_{0}}\left(N,\frac{\delta}{3L}\right)}{\underline{f}^{2}}+\frac{4\eta\;\mathcal{U}\operatorname{Rate}_{f_{0}}\left(N,\frac{\delta}{3L}\right)}{\underline{f}^{2}}.
				\end{aligned}
			\end{equation}
		}\noindent
			
		Next, we bound Eqn. \eqref{eqt:bound 2 DR 2}. By the orthogonality identity
		\eqref{eq:orth_propensity_residual}, conditioning on $\mathcal I_{-\ell}$ yields
		\begin{equation}
		\begin{aligned}
		\label{eqt:temp 2 final result of stochastic bound DR-2}
		&\left|
		\mathbb E\left[
		\mathbf 1_{\{A=\pi(X)\}}
		\big(\mathcal Y^{-1}-\hat m_0^\ell(\pi(X),X)\big)(t)
		\big(\hat g_0^\ell(A,X)-g(A,X)\big)
		\mid \mathcal I_{-\ell}
		\right]
		\right|
		\\
		&=
		\left|
		\mathbb E\left[
		\mathbf 1_{\{A=\pi(X)\}}
		\big(m_0(\pi(X),X)-\hat m_0^\ell(\pi(X),X)\big)(t)
		\big(\hat g_0^\ell(A,X)-g(A,X)\big)
		\mid \mathcal I_{-\ell}
		\right]
		\right|
		\\
		&\le
		\mathbb E\left[
		\left|
		\big(m_0(\pi(X),X)-\hat m_0^\ell(\pi(X),X)\big)(t)
		\right|
		\,
		\left|
		\hat g_0^\ell(A,X)-g(A,X)
		\right|
		\mid \mathcal I_{-\ell}
		\right]
		\\
		&\le
		\frac{
		\|\hat m_0^\ell-m_0\|_{\infty,[0,1]}
		\|\hat f_0^\ell-f_0\|_{\infty}
		}{\underline f^2}.
		\end{aligned}
		\end{equation}
		Consequently, on the nuisance event in Theorem~\ref{thm:statistical guarantee DR},
		\begin{equation}
		\label{eqt:temp 2 final result of stochastic bound DR-2-rate}
		\text{Eqn. \eqref{eqt:bound 2 DR 2}}
		\le
		\frac{
		\operatorname{Rate}_{f_{0}}\left(N,\frac{\delta}{3L}\right)
		\operatorname{Rate}_{m_{0}}\left(N,\frac{\delta}{3L}\right)
		}{\underline f^2}.
		\end{equation}
		Combining Eqns. \eqref{eqt:temp 2 final result of stochastic bound DR-1} and
		\eqref{eqt:temp 2 final result of stochastic bound DR-2-rate}, we have, with probability at least \(1-\delta\),
		\begin{equation}
		\begin{aligned}
		\label{eqt:final result of stochastic bound DR-2}
		&\text{Eqn. \eqref{eqt:final result 2 no expectation-DR}}
		\leq
		\text{Eqn. \eqref{eqt:bound 2 DR 1}}
		+
		\text{Eqn. \eqref{eqt:bound 2 DR 2}}
		\\
		&\leq
		\frac{2\operatorname{Rate}_{f_{0}}\left(N,\frac{\delta}{3L}\right)
		\bigg(2M+\operatorname{Rate}_{m_{0}}\left(N,\frac{\delta}{3L}\right)\bigg)}
		{\underline{f}^{2}}
		\\
		&\quad\times
		\left(
		\sqrt{\frac{2L\operatorname{N\text{-}dim}(\Pi)\log\left(e\frac{N}{L}d\right)}{N}}
		+
		\sqrt{\frac{2L\log\left(\frac{6\mathcal{J}L}{\delta}\right)}{N}}
		\right)
		\\
		&\quad+
		\frac{5\operatorname{Rate}_{f_{0}}\left(N,\frac{\delta}{3L}\right)
		\operatorname{Rate}_{m_{0}}\left(N,\frac{\delta}{3L}\right)}
		{\underline{f}^{2}}
		+
		\frac{4\eta\mathcal{U}\operatorname{Rate}_{f_{0}}\left(N,\frac{\delta}{3L}\right)}
		{\underline{f}^{2}}.
		\end{aligned}
		\end{equation}
		
		We also bound Eqn. \eqref{eqt:final result 3 no expectation-DR} in a similar manner. Indeed, 
		{\small
			\begin{align}
				&\text{Eqn. \eqref{eqt:final result 3 no expectation-DR}}\nonumber\\
				&=\underset{\pi\in\Pi}{\sup}\left|
				\frac{1}{L}\;\underset{\ell=1}{\overset{L}{\sum}}\;\mathbb{P}_{n,\ell}\left(\mathbf{1}_{\{A=\pi(X)\}}g(A,X)-1\right)\bigg(m_{0}(\pi(X),X)-\hat{m}_{0}^{\ell}(\pi(X),X)\bigg)(t)
				\right|\nonumber\\
				&\leq \frac{1}{L}\;\underset{\ell=1}{\overset{L}{\sum}}\;\underset{\pi\in\Pi}{\sup}\left|\mathbb{P}_{n,\ell}\left(
				\begin{aligned}
					&\left(\mathbf{1}_{\{A=\pi(X)\}}g(A,X)-1\right)\bigg(m_{0}(\pi(X),X)-\hat{m}_{0}^{\ell}(\pi(X),X)\bigg)(t)\\
					&-\mathbb{E}\left[\left(\mathbf{1}_{\{A=\pi(X)\}}g(A,X)-1\right)\bigg(m_{0}(\pi(X),X)-\hat{m}_{0}^{\ell}(\pi(X),X)\bigg)(t)\right]
				\end{aligned}
				\right)
				\right|\label{eqt:bound 3 DR 1}\\
				&\quad+\frac{1}{L}\;\underset{\ell=1}{\overset{L}{\sum}}\;\underset{\pi\in\Pi}{\sup}\left|
				\mathbb{E}\left[\left(\mathbf{1}_{\{A=\pi(X)\}}g(A,X)-1\right)\bigg(m_{0}(\pi(X),X)-\hat{m}_{0}^{\ell}(\pi(X),X)\bigg)(t)\right]\right|.\label{eqt:bound 3 DR 2}
			\end{align} 
		}\noindent
	
		We bound Eqn. \eqref{eqt:bound 3 DR 1}. Let
		\[
		\Delta_\ell(a,x,t):=m_0(a,x)(t)-\hat m_0^\ell(a,x)(t),
		\qquad
		W_\pi(X,A):=\mathbf 1_{\{A=\pi(X)\}}g(A,X)-1.
		\]
		By the fold-conditional orthogonality identity \eqref{eq:orth_outcome_residual}, for every fixed $(\pi,t)$,
		\[
		\mathbb E\left[W_\pi(X,A)\Delta_\ell(\pi(X),X,t)\mid \mathcal I_{-\ell}\right]=0.
		\]
		Thus the centered empirical process in Eqn. \eqref{eqt:bound 3 DR 1} can be written as
		\[
		\Gamma_3^\ell(t)
		:=
		\sup_{\pi\in\Pi}
		\left|
		\mathbb P_{n,\ell}\left[W_\pi(X,A)\Delta_\ell(\pi(X),X,t)\right]
		\right|.
		\]
		For the growth-function argument below, we first condition on the nuisance-training folds and on the held-out covariates; to avoid overloading notation, this conditioning is still written as conditioning on $\mathcal I_{-\ell}$. Let $n_\ell:=|\mathcal I_\ell|=N/L$ and define the set of labeling patterns realized by $\Pi$ on the held-out fold,
		\[
		\mathcal V_\ell
		:=
		\left\{(\pi(X_r))_{r\in\mathcal I_\ell}:\pi\in\Pi\right\}.
		\]
		For a fixed pattern $v=(v_r)_{r\in\mathcal I_\ell}\in\mathcal V_\ell$, define
		\[
		Z_{r,v_r}(t)
		:=
		\left(\mathbf 1_{\{A_r=v_r\}}g(v_r,X_r)-1\right)\Delta_\ell(v_r,X_r,t).
		\]
		Then
		\[
		\mathbb E\left[Z_{r,v_r}(t)\mid \mathcal I_{-\ell}\right]=0,
		\qquad
		|Z_{r,v_r}(t)|
		\le
		\left(1+\frac1{\underline f}\right)
		\|\hat m_0^\ell-m_0\|_{\infty,[0,1]}
		=:B_\ell.
		\]
		Using Hoeffding's inequality, for any fixed $v\in\mathcal V_\ell$, fixed $t$, and any $\varepsilon>0$,
		\[
		\mathbb P\left\{
		\left|
		\frac1{n_\ell}\sum_{r\in\mathcal I_\ell}Z_{r,v_r}(t)
		\right|
		\ge \varepsilon
		\,\middle|\,\mathcal I_{-\ell}
		\right\}
		\le
		2\exp\left(-\frac{n_\ell\varepsilon^2}{2B_\ell^2}\right).
		\]
		Since $|\mathcal V_\ell|\le m_\Pi(n_\ell)$, a union bound over the policy-induced labeling patterns gives
		\[
		\mathbb P\left\{\Gamma_3^\ell(t)\ge\varepsilon\mid\mathcal I_{-\ell}\right\}
		\le
		2m_\Pi(n_\ell)\exp\left(-\frac{n_\ell\varepsilon^2}{2B_\ell^2}\right).
		\]
		Choosing the right-hand side to be $\delta$ yields, for each fixed $t\in[0,1]$, with probability at least $1-\delta$,
		\[
		\Gamma_3^\ell(t)
		\le
		\left(1+\frac1{\underline f}\right)
		\|\hat m_0^\ell-m_0\|_{\infty,[0,1]}
		\sqrt{\frac{2\log\left(2m_\Pi(n_\ell)/\delta\right)}{n_\ell}}.
		\]
Using the multi-class Sauer bound and $n_\ell=N/L$, set
$u=\delta/(4L\mathcal J)$ in the preceding fixed-$t$ bound. Then, for any
fixed fold-grid pair $(\ell,j)$, with probability at least $1-u$,
\[
\Gamma_3^\ell(t_j)
\le
2\left(1+\frac1{\underline f}\right)
\|\hat m_0^\ell-m_0\|_{\infty,[0,1]}
\left(
\sqrt{\frac{2L\operatorname{N\text{-}dim}(\Pi)\log\left(e\frac NL d\right)}{N}}
+
\sqrt{\frac{2L\log\left(8\mathcal J L/\delta\right)}{N}}
\right).
\]
A union bound over all $L\mathcal J$ fold-grid pairs implies that the same
bound holds simultaneously for all $\ell=1,\ldots,L$ and all grid points
$j=1,\ldots,\mathcal J$, with probability at least $1-\delta/4$.
		Next, for $s,t\in[0,1]$, the elementary inequality
		\[
		\left|\sup_{\pi}|u_\pi|-\sup_{\pi}|v_\pi|\right|
		\le
		\sup_{\pi}|u_\pi-v_\pi|
		\]
		implies
		\[
		|\Gamma_3^\ell(t)-\Gamma_3^\ell(s)|
		\le
		\sup_{\pi\in\Pi}
		\left|
		\mathbb P_{n,\ell}\left[
		W_\pi(X,A)\{\Delta_\ell(\pi(X),X,t)-\Delta_\ell(\pi(X),X,s)\}
		\right]
		\right|.
		\]
		For the grid-to-continuum step, we use the fitted-Lipschitz condition maintained in Theorem~\ref{thm:statistical guarantee DR}, so that
		\[
		\sup_{a,x}\left|\Delta_\ell(a,x,t)-\Delta_\ell(a,x,s)\right|
		\le 2\mathcal U|t-s|.
		\]
		In particular, the deterministic envelope bound gives
		\[
		|\Gamma_3^\ell(t)-\Gamma_3^\ell(s)|
		\le
		2\left(1+\frac1{\underline f}\right)\mathcal U|t-s|.
		\]
		On this fixed-grid event, Lemma~\ref{lemma:sup all bound} and the deterministic
		$t$-oscillation bound below imply that, for every $\ell=1,\ldots,L$,
		\[
		\sup_{t\in[0,1]}\Gamma_3^\ell(t)
		\le
		2\left(1+\frac1{\underline f}\right)
		\|\hat m_0^\ell-m_0\|_{\infty,[0,1]}
		\left(
		\sqrt{\frac{2L\operatorname{N\text{-}dim}(\Pi)\log\left(e\frac NL d\right)}{N}}
		+
		\sqrt{\frac{2L\log\left(8\mathcal J L/\delta\right)}{N}}
		\right)
		+
		2\left(1+\frac1{\underline f}\right)\mathcal U\eta.
		\]
Intersecting the preceding fixed-grid event with the nuisance event
\[
\max_{1\le \ell\le L}
\|\hat m_0^\ell-m_0\|_{\infty,[0,1]}
\le
\operatorname{Rate}_{m_0}\left(N,\frac{\delta}{3L}\right),
\]
and using a union bound over the nuisance folds, we obtain, with probability
at least $1-\delta$,
		\begin{equation}
		\begin{aligned}
		\label{eqt:temp 3 final result of stochastic bound DR-1}
		\text{Eqn. \eqref{eqt:bound 3 DR 1}}
		& \leq \frac1L\sum_{\ell=1}^{L}\Gamma_3^\ell(t)
		\le
		\frac1L\sum_{\ell=1}^{L}\sup_{t\in[0,1]}\Gamma_3^\ell(t)
		\\
		&\le
		2\left(1+\frac1{\underline f}\right)
		\operatorname{Rate}_{m_{0}}\left(N,\frac{\delta}{3L}\right)
		\left(
		\sqrt{\frac{2L\operatorname{N\text{-}dim}(\Pi)\log\left(e\frac NL d\right)}{N}}
		+
		\sqrt{\frac{2L\log\left(8\mathcal J L/\delta\right)}{N}}
		\right)
		\\
		&\quad+
		2\left(1+\frac1{\underline f}\right)\mathcal U\eta.
		\end{aligned}
		\end{equation}
	
		Next, we bound Eqn. \eqref{eqt:bound 3 DR 2}. By the orthogonality identity
		\eqref{eq:orth_outcome_residual}, conditioning on $\mathcal I_{-\ell}$ gives
		\begin{equation}
		\begin{aligned}
		\label{eqt:temp 3 final result of stochastic bound DR-2}
		&\mathbb E\left[
		\left(\mathbf 1_{\{A=\pi(X)\}}g(A,X)-1\right)
		\big(m_0(\pi(X),X)-\hat m_0^\ell(\pi(X),X)\big)(t)
		\mid \mathcal I_{-\ell}
		\right]
		\\
		&=
		\mathbb E\left[
		\mathbb E\left[
		\mathbf 1_{\{A=\pi(X)\}}g(A,X)-1
		\mid X
		\right]
		\big(m_0(\pi(X),X)-\hat m_0^\ell(\pi(X),X)\big)(t)
		\mid \mathcal I_{-\ell}
		\right]
		\\
		&=0.
		\end{aligned}
		\end{equation}
		Therefore,
		\begin{equation}
		\label{eqt:temp 3 final result of stochastic bound DR-2-zero}
		\text{Eqn. \eqref{eqt:bound 3 DR 2}}=0.
		\end{equation}
		Combining Eqns. \eqref{eqt:temp 3 final result of stochastic bound DR-1} and
		\eqref{eqt:temp 3 final result of stochastic bound DR-2-zero}, we have, with probability at least \(1-\delta\),
		\begin{equation}
		\begin{aligned}
		\label{eqt:final result of stochastic bound DR-3}
		&\text{Eqn. \eqref{eqt:final result 3 no expectation-DR}}
		\leq
		\text{Eqn. \eqref{eqt:bound 3 DR 1}}
		+
		\text{Eqn. \eqref{eqt:bound 3 DR 2}}
		\\
		&\leq
		2\left(1+\frac{1}{\underline{f}}\right)
		\operatorname{Rate}_{m_{0}}\left(N,\frac{\delta}{3L}\right)
		\left(
		\sqrt{\frac{2L\operatorname{N\text{-}dim}(\Pi)\log\left(e\frac{N}{L}d\right)}{N}}
		+
		\sqrt{\frac{2L\log\left(\frac{8\mathcal{J}L}{\delta}\right)}{N}}
		\right)
		\\
		&\quad+
		2\left(1+\frac{1}{\underline{f}}\right)\mathcal U\eta.
		\end{aligned}
		\end{equation}
		
		\paragraph{Step III: Studying Eqn. \eqref{eqt:final result 4 no expectation-DR}}
		Recall that
		\begin{equation*}
			\begin{aligned}
				\mathbb{E}[\mathcal{Y}^{-1}[\pi(X)]](t)=\mathbb{E}\left[m_{0}(\pi(X),X)+\frac{\mathbf{1}_{\{A=\pi(X)\}}}{f_{0}(A|X)}(\mathcal{Y}^{-1}-m_{0}(\pi(X),X))\right](t),
			\end{aligned}
		\end{equation*}
		we therefore have
		\begin{equation*}
			\begin{aligned}
				&\underset{\pi\in\Pi}{\sup}\left|\mathbb{P}_{N}\left\{m_{0}(\pi(X),X)+\frac{\mathbf{1}_{\{A=\pi(X)\}}}{f_{0}(A|X)}(\mathcal{Y}^{-1}-m_{0}(\pi(X),X))\right\}(t)-\mathbb{E}[\mathcal{Y}^{-1}[\pi(X)]](t)\right|\\
				&=\underset{\pi\in\Pi}{\sup}\left|
				\begin{aligned}
					&\mathbb{P}_{N}\left\{m_{0}(\pi(X),X)+\frac{\mathbf{1}_{\{A=\pi(X)\}}}{f_{0}(A|X)}(\mathcal{Y}^{-1}-m_{0}(\pi(X),X))\right\}(t)\\
					&-\mathbb{E}\left[m_{0}(\pi(X),X)+\frac{\mathbf{1}_{\{A=\pi(X)\}}}{f_{0}(A|X)}(\mathcal{Y}^{-1}-m_{0}(\pi(X),X))\right](t)
				\end{aligned}
				\right|:=\bigtriangleup(t).
			\end{aligned}
		\end{equation*}
		Also, denote
		\begin{equation*}
			\begin{aligned}
				&\left|
				\begin{aligned}
					&\mathbb{P}_{N}\left\{m_{0}(a_{i},X)+\frac{\mathbf{1}_{\{A=a_{i}\}}}{f_{0}(A|X)}(\mathcal{Y}^{-1}-m_{0}(a_{i},X))\right\}(t)\\
					&-\mathbb{E}\left[m_{0}(a_{i},X)+\frac{\mathbf{1}_{\{A=a_{i}\}}}{f_{0}(A|X)}(\mathcal{Y}^{-1}-m_{0}(a_{i},X))\right](t)
				\end{aligned}
				\right|:=\bigtriangleup^{i}(t).
			\end{aligned}
		\end{equation*}
		Using Hoeffding's inequality, we have: for any $\varepsilon>0$,
		\begin{equation*}
			\mathbb{P}\left\{\bigtriangleup^{i}(t)\geq \varepsilon\big| X_{1:N}\right\}\leq 2\exp\left(-\frac{N\varepsilon^2}{2B^2}\right),
		\end{equation*}
		where \(|m_{0}(a_{i},X)+\frac{\mathbf{1}_{\{A=a_{i}\}}}{f_{0}(A|X)}(\mathcal{Y}^{-1}-m_{0}(a_{i},X))(t)|\leq B\) such that we can choose \(B\) as \((1+\frac{2}{\underline{f}})M\). Since there are at most $m_{\Pi}(N)$ distinct labeling patterns realized by $\Pi$ on $X_{1:N}$. Applying a union bound over these patterns yields
		\begin{equation*}
			\mathbb{P}\left\{\bigtriangleup(t)\geq\varepsilon\big| X_{1:N}\right\}\leq 2m_{\Pi}(N)\exp\left(-\frac{N\varepsilon^2}{2B^2}\right).
		\end{equation*}
		Choose $\varepsilon$ so that the RHS equals $\delta$:
		\begin{equation*}
			2\,m_{\Pi}(N)\exp\!\left(-\frac{N\varepsilon^2}{2B^{2}}\right)=\delta\quad\Longleftrightarrow\quad\varepsilon=B\sqrt{\frac{2\log\left(\frac{2m_{\Pi}(N)}{\delta}\right)}{N}}.
		\end{equation*}
		As a result, for each fixed \(t\in[0,1]\), with probability at least \(1-\delta\), 
		\begin{equation*}
			\begin{aligned}
				\bigtriangleup(t)\leq \left(1+\frac{2}{\underline{f}}\right)M\sqrt{\frac{2\log\left(\frac{2m_{\Pi}(N)}{\delta}\right)}{N}}.
			\end{aligned}
		\end{equation*}
		In particular, for each fixed grid point \(t_j\), with probability at least \(1-\frac{\delta}{\mathcal{J}}\),
		\begin{equation*}
			\begin{aligned}
				\bigtriangleup(t_{j})&\leq 2\left(1+\frac{2}{\underline{f}}\right)M\left(\sqrt{\frac{2L\operatorname{N\text{-}dim}(\Pi)\log\left(e\frac{N}{L}d\right)}{N}}+\sqrt{\frac{2L\log\left(\frac{2\mathcal{J}}{\delta}\right)}{N}}\right).
			\end{aligned}
		\end{equation*}
		Note that
        {\small
		\begin{align*}
			& \bigtriangleup(t)-\bigtriangleup(s)\\
			&\leq\underset{\pi\in\Pi}{\sup}\left|
			\begin{aligned}
				&\mathbb{P}_{N}\left\{m_{0}(\pi(X),X)+\frac{\mathbf{1}_{\{A=\pi(X)\}}}{f_{0}(A|X)}(\mathcal{Y}^{-1}-m_{0}(\pi(X),X))\right\}(t)\\
				&-\mathbb{P}_{N}\left\{m_{0}(\pi(X),X)+\frac{\mathbf{1}_{\{A=\pi(X)\}}}{f_{0}(A|X)}(\mathcal{Y}^{-1}-m_{0}(\pi(X),X))\right\}(s)
			\end{aligned}
			\right|\\
			&\quad+\underset{\pi\in\Pi}{\sup}\left|
			\begin{aligned}
				&\mathbb{E}\left[m_{0}(\pi(X),X)+\frac{\mathbf{1}_{\{A=\pi(X)\}}}{f_{0}(A|X)}(\mathcal{Y}^{-1}-m_{0}(\pi(X),X))\right](s)\\
				&-\mathbb{E}\left[m_{0}(\pi(X),X)+\frac{\mathbf{1}_{\{A=\pi(X)\}}}{f_{0}(A|X)}(\mathcal{Y}^{-1}-m_{0}(\pi(X),X))\right](t)
			\end{aligned}
			\right|\\
			&\leq 2\left(1+\frac{2}{\underline{f}}\right)\mathcal{U}|t-s|,
		\end{align*}
        }\noindent
		implying that \(|\bigtriangleup(t)-\bigtriangleup(s)|\leq 2\left(1+\frac{2}{\underline{f}}\right)\mathcal{U}|t-s|\) by symmetrization. By Lemma \ref{lemma:sup all bound}, we then have, with probability at least \(1-\delta\),
        {\small
		\begin{equation}
					\begin{aligned}\label{eqt:final result of stochastic bound DR-4}
				\underset{t\in[0,1]}{\sup}\;\bigtriangleup(t)\leq&2\left(1+\frac{2}{\underline{f}}\right)M\left(\sqrt{\frac{2L\operatorname{N\text{-}dim}(\Pi)\log\left(e\frac{N}{L}d\right)}{N}}+\sqrt{\frac{2L\log\left(\frac{2\mathcal{J}}{\delta}\right)}{N}}\right)\\
				&\quad +\left(1+\frac{2}{\underline{f}}\right)\mathcal{U}\eta.
			\end{aligned}
		\end{equation}
        }\noindent

		\paragraph{Final combination.}
		Define
		\[
		r_f:=\operatorname{Rate}_{f_{0}}\left(N,\frac{\delta}{3L}\right),
		\qquad
		r_m:=\operatorname{Rate}_{m_{0}}\left(N,\frac{\delta}{3L}\right).
		\]
		Without loss of generality, the rate functions are taken to be nonincreasing in the confidence parameter after replacing them by their monotone envelopes. Also define
		\[
		\mathcal{V}_N(\Pi,\delta)
		:=
		\sqrt{\frac{2L\operatorname{N\text{-}dim}(\Pi)\log\left(e\frac{N}{L}d\right)}{N}}
		+
		\sqrt{\frac{2L\log\left(\frac{8\mathcal{J}L}{\delta}\right)}{N}}.
		\]
		Combining Eqns.~\eqref{eqt:final result of stochastic bound DR-2},
		\eqref{eqt:final result of stochastic bound DR-3}, and
		\eqref{eqt:final result of stochastic bound DR-4}, we have, with probability at least \(1-4\delta\),
		{\small
		\begin{align*}
		\mathcal{R}(\hat{\pi}^{\mathrm{DR}})
		&\leq
		2L_{U}
		\left[
		2M\left(1+\frac{2}{\underline f}\right)
		\mathcal{V}_N(\Pi,\delta)
		+
		\left(4+\frac{6}{\underline f}\right)\mathcal U\eta
		\right]
		\\
		&\quad+
		2L_U
		\Bigg[
		\left(
		\frac{2r_fr_m}{\underline f^2}
		+
		\frac{4Mr_f}{\underline f^2}
		+
		2\left(1+\frac{1}{\underline f}\right)r_m
		\right)
		\mathcal{V}_N(\Pi,\delta)
		\\
		&\qquad\qquad
		+
		\frac{5r_fr_m}{\underline f^2}
		+
		\frac{4\mathcal U r_f}{\underline f^2}\eta
		\Bigg].
		\end{align*}
		}\noindent
		This matches \eqref{eq:dr_regret_rearranged}.

	\end{proof}

\section{Proof of Theorem~\ref{thm:lower_bound_main}}\label{app:proof_lower_bound}

\begin{proof}
	Let $V:=\mathrm{N\text{-}dim}(\Pi)\ge 1$.
	By definition of the Natarajan dimension, there exists a set $S=\{x_1,\dots,x_V\}\subset\mathcal X$ that is Natarajan-shattered by $\Pi$.
	Hence there exist two functions $f_1,f_2:S\to\mathcal A$ with $f_1(x_i)\neq f_2(x_i)$ such that for every subset $S_0\subseteq S$ there exists $\pi\in\Pi$ with
	$\pi(x)=f_1(x)$ for $x\in S_0$ and $\pi(x)=f_2(x)$ for $x\in S\setminus S_0$.
	
	Fix $\delta_0\in(0,1/4)$ (to be chosen later) and let $\mathcal V=\{\pm 1\}^V$.
	We construct a finite subset $\{\mathbb P_v\}_{v\in\mathcal V}\subset\mathcal P_{\text{lower}}(q_-,q_+)$ such that learning the optimal policy reduces to identifying the vertex $v$.
	
	First, we construct distribution-valued potential outcomes. Under the distribution $\mathbb P_v$:
	\begin{itemize}
		\item $X$ is uniformly distributed on $S$.
		\item The behavior policy satisfies, for each $x_i\in S$,
		\[
		f_0(f_2(x_i)\mid x_i)=\underline f,
		\qquad
		f_0(a\mid x_i)=\underline f\quad\text{for }a\notin\{f_1(x_i)\},
		\]
		and
		\[
		f_0(f_1(x_i)\mid x_i)=1-(d-1)\underline f.
		\]
		\item Potential outcomes are distribution-valued and take only the two valid base quantile curves $q_-$ and $q_+$ from Theorem~\ref{thm:lower_bound_main}.
		Let $\Delta(t):=q_+(t)-q_-(t)\ge 0$.
		For each $i$ and each action $a\in\mathcal A$, define a Bernoulli latent variable $Z_{i,a}\in\{0,1\}$, drawn conditionally on $X=x_i$ and independently of $A$, and set
		\[
		\mathcal Y[a]^{-1}(t)\ \big|\ (X=x_i)\ :=\ q_-(t)+Z_{i,a}\,\Delta(t),\qquad t\in[0,1].
		\]
		We choose the success probabilities:
		\[
		Z_{i,f_1(x_i)}\sim \mathrm{Bern}(1/2),\qquad
		Z_{i,f_2(x_i)}\sim \mathrm{Bern}(1/2+v_i\delta_0),
		\qquad
		Z_{i,a}\sim \mathrm{Bern}(0)\ \text{otherwise}.
		\]
		\item The observed outcome satisfies $\mathcal Y=\mathcal Y[A]$ (consistency).
	\end{itemize}
	This construction satisfies Assumptions~\ref{ass:consistency}--\ref{ass:quantile_lipschitz}: unconfoundedness holds because the latent variables are conditionally independent of $A$ given $X$, boundedness and quantile regularity follow from the properties of $q_-$ and $q_+$, and the overlap constant is $\underline f$ by design since $0<\underline f\le 1/d$.
	
	Now we can evaluate $U_\alpha$ and identify the optimal policy.
	Under $\mathbb P_v$, let $\mu_v(\pi)$ denote the policy-induced barycenter quantile (which depends on the distribution $\mathbb P_v$).
	By Proposition~\ref{prop:barycenter_quantile_main}:
	\[
	\mu_v(\pi)^{-1}(t)=\mathbb E_{(X,\mathcal Y)\sim \mathbb P_v}\big[\mathcal Y[\pi(X)]^{-1}(t)\big].
	\]
	For each $x_i$, choosing $f_2(x_i)$ instead of $f_1(x_i)$ changes the mean quantile by $\delta_0\,v_i\,\Delta(t)$.
	Hence the optimal policy $\pi_v^\star$ for $\mathbb P_v$ satisfies
	\[
	\pi_v^\star(x_i)=
	\begin{cases}
		f_2(x_i), & v_i=+1,\\
		f_1(x_i), & v_i=-1,
	\end{cases}
	\]
	and for any $\pi\in\Pi$,
	\[
	U_\alpha(\mu_v(\pi_v^\star))-U_\alpha(\mu_v(\pi))
	\ge
	\frac{\delta_0}{V}\left(\sum_{i=1}^V \mathbf 1\{\pi(x_i)\neq \pi_v^\star(x_i)\}\right)\int_0^\alpha \Delta(t)\,dt.
	\]
	
	Let $\mathbb P_v^{N}$ denote the joint distribution of the dataset $\mathcal D_N$ (containing $N$ logged samples) under $\mathbb P_v$.
	Recall that the estimator $\hat\pi$ is a function of $\mathcal D_N$.
	\[
	\sup_{v\in\mathcal V}\mathbb E_{\mathbb P_v^N}\!\left[U_\alpha(\mu_v(\pi_v^\star))-U_\alpha(\mu_v(\hat\pi))\right]
	\ge
	\frac{1}{2^V}\sum_{v\in\mathcal V}\mathbb E_{\mathbb P_v^N}\!\left[U_\alpha(\mu_v(\pi_v^\star))-U_\alpha(\mu_v(\hat\pi))\right].
	\]
	
	Therefore, based on the above results, we have
	\[
	\frac{1}{2^V}\sum_{v\in\mathcal V}\mathbb E_{\mathbb P_v^N}\!\left[U_\alpha(\mu_v(\pi_v^\star))-U_\alpha(\mu_v(\hat\pi))\right]
	\ge
	\frac{\delta_0}{V \cdot 2^V} \sum_{v\in\mathcal V} \sum_{i=1}^V \mathbb P_v^N\big(\hat\pi(x_i) \neq \pi_v^\star(x_i)\big) \int_0^\alpha \Delta(t)\,dt.
	\]
	We now apply Assouad's symmetrization. For each $i \in \{1,\dots,V\}$, let $M_i[v]$ be the vertex in $\mathcal V$ that differs from $v$ only in the $i$-th coordinate. By swapping the order of summation and pairing each $v$ with $M_i[v]$ for $v_i=1$, we obtain
	\begin{align*}
	\sum_{v\in\mathcal V} \mathbb P_v^N\big(\hat\pi(x_i) \neq \pi_v^\star(x_i)\big) 
	&= \sum_{v: v_i=1} \Big( \mathbb P_v^N\big(\hat\pi(x_i) \neq f_2(x_i)\big) + \mathbb P_{M_i[v]}^N\big(\hat\pi(x_i) \neq f_1(x_i)\big) \Big)\\
    &\geq\sum_{v: v_i=1} \Big( \mathbb P_v^N\big(\hat\pi(x_i) \neq f_2(x_i)\big) + 1 - \mathbb P_{M_i[v]}^N\big(\hat\pi(x_i) \neq f_2(x_i)\big) \Big)\\
    &\geq\sum_{v: v_i=1} \Big(1 - \operatorname{TV}(\mathbb P_{v}^N,\mathbb P_{M_i[v]}^N)\Big),    
	\end{align*}
where \(\operatorname{TV}(\mathbb{P},\mathbb{Q})\) represents the total variation distance between \(\mathbb{P}\) and \(\mathbb{Q}\). Applying this lower bound to the sum yields
\[
\frac{1}{2^V}\sum_{v\in\mathcal V}\mathbb E_{\mathbb P_v^N}\!\left[U_\alpha(\mu_v(\pi_v^\star))-U_\alpha(\mu_v(\hat\pi))\right]
\ge
\frac{\delta_0}{V}\left(\frac{1}{2^V}\sum_{i=1}^V\sum_{\substack{v\in\mathcal V\\ v_i=1}}
\big(1-\mathrm{TV}(\mathbb P_v^{N},\mathbb P_{M_i[v]}^{N})\big)\right)\int_0^\alpha \Delta(t)\,dt.
\]
According to the relationship between TV distance and KL divergence (e.g., Lemma 2.6 in \cite{tsybakov2008nonparametric}), we have: for any distributions $\mathbb{P},\;\mathbb{Q}$,
\[
1-\mathrm{TV}(\mathbb{P},\mathbb{Q})\ \ge\ \frac12 \exp\big(-\mathrm{KL}(\mathbb{P}\|\mathbb{Q})\big).
\]
In our construction, $\mathbb P_v$ and $\mathbb P_{M_i[v]}$ differ only when $(X=x_i,\,A=f_2(x_i))$,
which happens with probability $\frac{1}{V}\underline f$.
Conditioned on this event, the observed outcome reveals the Bernoulli latent variable
$Z_{i,f_2(x_i)}\sim \mathrm{Bern}(1/2+\delta_0)$ versus $\mathrm{Bern}(1/2-\delta_0)$.
A standard bound for Bernoulli KL implies that for $\delta_0\in(0,1/4)$,
\[
\mathrm{KL}(\mathbb P_v^{N}\|\mathbb P_{M_i[v]}^{N})
\le
\frac{12\underline f N}{V}\,\delta_0^2.
\]
Choosing
\[
\delta_0:=\frac{1}{\sqrt{24}}\min\left\{1,\sqrt{\frac{V}{\underline f N}}\right\}
\]
ensures $\delta_0\in(0,1/4)$ and
\[
\frac{12\underline f N}{V}\delta_0^2
\le \frac12.
\]
Hence $1-\mathrm{TV}(\mathbb P_v^{N},\mathbb P_{M_i[v]}^{N})\ge \frac12 e^{-1/2}$.

Plugging this bound into the preceding display yields
\[
\sup_{v\in\mathcal V}\mathbb E_{\mathbb P_v^N}\!\left[U_\alpha(\mu_v(\pi_v^\star))-U_\alpha(\mu_v(\hat\pi))\right]
\ge
\frac{\delta_0}{4}e^{-1/2}\int_0^\alpha \Delta(t)\,dt
=
\frac{e^{-1/2}}{4\sqrt{24}}
\min\left\{1,\sqrt{\frac{V}{\underline f N}}\right\}
\int_0^\alpha (q_+(t)-q_-(t))\,dt.
\]
Since $\{\mathbb P_v\}_{v\in\mathcal V}\subset\mathcal P_{\text{lower}}(q_-,q_+)$, this lower bound holds for $\sup_{\mathbb P\in\mathcal P_{\text{lower}}(q_-,q_+)}$ as well.
This proves the theorem with $c_0=e^{-1/2}/(4\sqrt{24})$.
\end{proof}

\end{document}